\documentclass[a4paper,UKenglish,cleveref, autoref, thm-restate]{lipics-v2021}

\usepackage{graphicx}
\usepackage{xcolor}
\usepackage{tikz}
\usetikzlibrary{
    arrows.meta,
    positioning,
    calc,
    automata,
    decorations.pathmorphing
}
\usepackage{booktabs}
\usepackage{array}
\usepackage{multirow}
\usepackage{todonotes}
\usepackage{xspace}

\newcommand{\Q}{\mathbb{Q}}

\newcommand{\Z}{\mathbb{Z}}
\newcommand{\R}{\mathbb{R}}

\newcommand{\TDS}{\textsf{TDS}\xspace}

\newcommand{\Paths}{\mathsf{Paths}} 
\newcommand{\FPaths}{\mathsf{Paths_{fin}}}
\newcommand{\trace}{\mathit{trace}}

\newcommand{\opt}{{\textsf{opt}}} 
\newcommand{\DS}{\mathrm{DS}_{\lambda}}

\renewcommand{\Pr}{\mathbb{P}} 

\newcommand{\FM}{\mathsf{FM}}
\newcommand{\Max}{\mathsf{Max}}
\newcommand{\Min}{\mathsf{Min}}

\newcommand{\BSCC}{\textsf{BSCC}\xspace}

\newcommand{\lang}[1]{L_{#1}} 
\newcommand{\langinf}[1]{L^\infty_{#1}} 

\usepackage{amsthm}

\newtheorem{problem}{Problem}

\hideLIPIcs  

\title{
The Stochastic Target Discounted-Sum Problem
} 

\author{Nathalie Bertrand}{Univ Rennes, Inria, CNRS, IRISA, France}{}{}{}
\author{Pranav Ghorpade}{University of Sydney, Australia}{}{}{}
\author{Senthil Rajasekaran}{Université Libre de Bruxelles, Belgium}{}{}{}
\author{Sasha Rubin}{University of Sydney, Australia}{}{}{}
\author{Moshe Y. Vardi}{Rice University, USA}{}{}{}

\authorrunning{N. Bertrand, P. Ghorpade, S. Rajasekaran, S. Rubin, M. Y. Vardi}

\Copyright{Jane Open Access and Joan R. Public} 



\category{} 

\relatedversion{} 

\acknowledgements{}

\nolinenumbers 

\begin{document}

\maketitle

\begin{abstract}
The \emph{target discounted-sum problem} (\TDS) asks, given a finite integer alphabet $\Sigma$, a rational discount factor $\lambda$, and a rational target $t$, whether some infinite sequence over $\Sigma$ has discounted sum exactly $t$. This problem remains open and underlies several open questions in automata theory, games, and Markov decision processes.

We introduce and solve its stochastic counterpart, the \emph{stochastic target discounted-sum problem}, which replaces existence by computation of the probability. We show that the probability that a random sequence generated by a finite Markov chain has discounted sum $t$ is rational and computable in pseudo-polynomial time. We further show how to decide, in polynomial time, whether the discounted-sum distribution of a Markov chain is atomless, and how to approximate to an arbitrary precision the probability that the discounted sum exceeds a rational threshold.

Our techniques for the stochastic \TDS problem allow us to make progress on \TDS objectives in stochastic games, which are known to be as hard as the \TDS problem. Restricting the maximizing player to finite-memory strategies, while allowing the minimizing player to use arbitrary strategies, we reduce the value problem and the synthesis problem to corresponding problems for safety objectives in stochastic games. This yields computable optimal values and deterministic optimal strategies with pseudo-polynomially bounded memory for stochastic games, and results in pseudo-polynomial-time algorithms for special cases of Markov decision processes and deterministic two-player games.
\end{abstract}

\section{Introduction}
\subparagraph*{Discounted sum and target discounted-sum problem}
Assessing the performance of a system often starts with assigning a numerical value to its executions. One way to do so, is through discounted-sum payoff, which aggregates an infinite sequence of rewards along an execution in a discounted sum~\cite{deAlfaro2003}. It is central to reinforcement learning~\cite{DBLP:books/lib/SuttonB98,DBLP:books/lib/BertsekasT96} hence optimization of Markov decision processes~\cite{Puterman94} and optimal strategies in games~\cite{10.1007/978-3-642-10631-6_13} and has also been studied from a language theoretic perspective for automata~\cite{DBLP:journals/corr/BokerH14, bertrand2026categorizerautomatadiscountedsumpayoffs, BCV22, DBLP:journals/lmcs/BokerH25}.

The discounted sum favors immediate rewards and mitigates the effect of future rewards. Formally, for a rational discount factor $\lambda \in (0,1)$ and a finite alphabet of rewards $\Sigma \subset \Z$, the discounted sum of an infinite sequence $d_0d_1d_2\cdots \in \Sigma^\omega$ is $\DS(d_0d_1\cdots) := \sum_{i=0}^{\infty} d_i \lambda^i$.  
Introduced by Raskin and later studied by Boker, Henzinger and Otop, the \emph{target discounted-sum problem} (\TDS)~\cite{boker_target_2015} asks, given $\lambda$, $\Sigma$ and a rational target $t$, whether some sequence over $\Sigma$ has discounted sum exactly~$t$. This simply stated problem is open, and constitutes the core difficulty in many open problems in automata theory, games, Markov decision processes, and number representations~\cite{boker_target_2015,DBLP:conf/fsttcs/HunterR14,DBLP:journals/fmsd/RandourRS17,chatterjee2017optimizing,filiot2020adversarial,brice2023rational}. 

\subparagraph*{The stochastic target discounted-sum problem}
In this paper we solve the stochastic counterpart of the target discounted-sum problem.
\begin{problem}[Stochastic Target Discounted-Sum]\label{prob:TDS}~\\
{\bf Given:} a finite alphabet $\Sigma \subseteq \Z$, a rational discount factor $\lambda \in (0,1)$, a rational target value $t$, and a finite Markov chain $C$ with a designated initial state, rational transition probabilities, and transitions labeled by elements of $\Sigma$.\\
{\bf Compute:} the probability that (the sequence of labels along) a random infinite path from the initial state has discounted sum equal to $t$.
\end{problem}
Our main contribution (\cref{thm:MC}) shows that \cref{prob:TDS} is solvable in pseudo-polynomial time. More precisely, if $\mu=\max_{d\in\Sigma}|d|$ is the largest absolute value in $\Sigma$, and $q_\lambda$ and $q_t$ the respective denominators of $\lambda$ and $t$, the algorithm runs in time polynomial in the size of the Markov chain and in the numerical values of $\mu$, $q_\lambda$, and $q_t$.

A natural approach in quantitative verification of Markov chains is to represent the set of behaviors of interest by a deterministic automaton, and then use a standard product construction to compute the probability that the Markov chain generates a run accepted by the automaton~\cite{PMC-book}. Applied to the stochastic \TDS problem, this amounts to considering the \TDS \emph{language}
$L_{[\Sigma,\lambda,t]}:=\{w\in\Sigma^\omega \mid \DS(w)=t\}$ and computing the
probability that $C$ generates a word in $L_{[\Sigma,\lambda,t]}$. 
When $1/\lambda$ is an integer, this
settles the problem immediately, since the \TDS language is then recognized by an
(effectively constructible) deterministic safety automaton~\cite{BCV22}. 
For arbitrary rational discount factors the \TDS language is in general not $\omega$-regular~\cite{BCV22}, so this approach does not apply directly.

To solve the stochastic \TDS problem for arbitrary rational discount factors, we show that the \TDS language contains an $\omega$-regular sub-language whose probability agrees with the one of the \TDS language in every finite Markov chain. To do so, our starting point is a notion of \emph{suffix sums}, implicit in the decidability proofs for special cases of the \TDS problem~\cite{boker_target_2015}.  Informally, the suffix sums of a sequence are the discounted sums of all its suffixes. Formally, for a sequence $d_0d_1d_2\dots$, these are the values $\DS(d_id_{i+1}\cdots)$ for every $i\geq 0$. 
Using this notion, we partition the \TDS language according to whether a sequence has finitely or infinitely many distinct suffix sums. Our key technical result is that, in every finite Markov chain, the probability of generating a sequence with given discounted sum and infinitely many distinct suffix sums is zero. 
Consequently, in every Markov chain, the probability of the \TDS language is equal to the probability of the sub-language of sequences with only finitely many distinct suffix sums, written $L^{<\infty}_{[\Sigma,\lambda,t]}$. We then show that $L^{<\infty}_{[\Sigma,\lambda,t]}$ can be recognized by an effectively constructible deterministic safety automaton of pseudo-polynomial size, giving a pseudo-polynomial-time algorithm for the stochastic target discounted-sum problem via the above-mentioned standard product construction~\cite{PMC-book}.

\subparagraph*{Distribution of discounted sums in Markov chains}
Beyond answering the stochastic \TDS problem, our techniques permit a finer analysis of the distribution of discounted-sum payoffs in finite Markov chains.
Classical algorithms for finite Markov chains with discounted-sum payoffs compute the expectation of the induced distribution~\cite{Puterman94}. The expected value reveals little about the distribution itself: for instance, it does not determine the probability mass at a given point, as in the stochastic \TDS problem.

To start better understanding the distribution of discounted sums, we give a polynomial-time algorithm to decide whether the cumulative distribution function of the discounted-sum payoff is continuous. Alternatively, this algorithm decides in polynomial time whether the distribution is atomless. In comparison to the stochastic \TDS, for which a target value is given as input, the question is whether the Markov chain admits no \emph{atom}, i.e. a target value with positive probability in the discounted-sum distribution. 

Second, we show how to approximate within any desired precision the probability that the discounted sum exceeds a threshold. A natural approach, and the one taken in~\cite{DBLP:journals/fmsd/RandourRS17}, is a conservative approximation obtained by truncation. Only the first $n$ steps of executions are considered while the tails are discarded, since their contribution to the discounted sum tends to zero as $n$ grows. The prefix then allows one to determine how the payoff compares to the threshold $t$, assuming the prefix value is far enough from $t$; otherwise the comparison is left undecided. This yields a lower and an upper bound on the probability that the payoff exceeds $t$. Although the lower bound converges to the desired probability, the gap between the bounds converges to the probability that the payoff equals $t$, which may be positive. Thus, truncation alone does not guarantee that a prescribed precision can be certified. We use our algorithm for the stochastic \TDS problem to compute the probability that the discounted sum equals threshold, and subtract it from the upper bound, so that the gap tends to zero, thus providing a stopping criterion that is guaranteed to be met for every precision $\varepsilon>0$.

\subparagraph*{\TDS objectives in stochastic games} $2$-player games with \TDS objectives were first introduced by Hunter and Raskin~\cite{DBLP:conf/fsttcs/HunterR14}, who left the problem open. 

For Markov decision processes (MDPs), (also referred to as $1\frac12$-player games, or 1-player stochastic games) with \TDS objectives,
we show that both the infimum probability over all strategies and the supremum probability over finite-memory strategies are rational and computable in pseudo-polynomial time. In both cases, the value is attained by a deterministic strategy with pseudo-polynomially bounded memory, which can be synthesized within the same time bound. 
Computing the unrestricted supremum remains open, since it subsumes the (non-stochastic) \TDS problem. 
We also show that the supremum under finite-memory strategies can be strictly smaller than the unrestricted supremum by giving a separating example in which an infinite-memory strategy attains probability $1$, while every finite-memory strategy attains probability $0$. Our results also yield a memory dichotomy: whenever a supremum-optimal strategy exists, either one exists with pseudo-polynomially bounded memory, or every such strategy requires infinite memory.

These results extend to stochastic games (also called $2\frac12$-player games) with \TDS objectives. We show that the value the maximizing player can guarantee using (unbounded) finite memory, against arbitrary strategies of the minimizing player, is computable and attained by a deterministic strategy with pseudo-polynomially bounded memory. $2$-player games are the special case of $2\frac12$-player games. In this setting, the value and an optimal strategy are computable in pseudo-polynomial time, settling the finite-memory restriction of the $2$-player games with \TDS objectives of~\cite{DBLP:conf/fsttcs/HunterR14}.

\subparagraph*{More related work}
In~\cite{boker_target_2015} the \TDS problem is related to, among other topics, $\beta$-expansions~\cite{Rnyi1957RepresentationsFR, Schmidt1980OnPE} and discounted-sum automata~\cite{DBLP:journals/corr/BokerH14,DBLP:journals/lmcs/BokerH25}, and decidability is established for several special cases. In particular, they show that it is decidable whether there exists an \emph{eventually periodic} word whose discounted sum is equal to the target. Apart from deterministic settings of quantitative games~\cite{DBLP:conf/fsttcs/HunterR14,brice2023rational,filiot2020adversarial} and discounted sum automata~\cite{DBLP:journals/corr/BokerH14,DBLP:journals/lmcs/BokerH25}, \TDS problem is also a source of hardness for several open problems in stochastic settings: multi-objective model checking in MDPs~\cite{chatterjee_multi-objective_2013}, model checking percentile queries in MDPs~\cite{DBLP:journals/fmsd/RandourRS17}, optimizing piecewise-constant utilities in MDPs~\cite{bertrand2026categorizerautomatadiscountedsumpayoffs}, and optimizing expectation with guarantees in POMDPs~\cite{chatterjee2017optimizing} are all known to be at least as hard as \TDS. This motivates our study of the stochastic \TDS problem and subsequently its application to \TDS objectives in MDPs.

Classical algorithms for finite Markov chains compute the expected discounted reward~\cite{Puterman94}. White~\cite{white_minimizing_1993} and Wu and Lin~\cite{wu_minimizing_1999} consider, in MDPs, the probability that the discounted reward lies above or below a given threshold. They establish structural results, including optimality equations and the existence of optimal policies under appropriate conditions, but leave the algorithmic treatment of these problems open even for Markov chains. 
Our results on Markov chains partially address these open problems, as they enable the approximation --up to any precision-- of the probability that the discounted reward lies above a given threshold.

\section{Preliminaries}\label{sec:prel}
 
\subparagraph*{Words}
Fix a finite \emph{alphabet} $\Sigma$. A \emph{word} (resp.\ \emph{finite word}) is
an infinite (resp.\ finite) sequence of symbols from $\Sigma$; we write $\Sigma^\omega$
and $\Sigma^*$ for the sets of words and of finite words, $|w|$ for the length of a
finite word $w$, and $\varepsilon$ for the empty word. Positions are numbered from $0$,
so that $w=w[0]w[1]w[2]\cdots$, and we write $w[0..i]:=w[0]\cdots w[i]$ for a prefix,
$w[i..]:=w[i]w[i+1]\cdots$ for a suffix, and $w\cdot w'$ for concatenation. We use the
same position notation for the sequences of states introduced below.

\subparagraph*{Automata}
A \emph{deterministic safety automaton} is a tuple $\mathcal A=(Q,\Sigma,q_\textsf{init},\delta,F)$,
where $Q$ is a finite set of states, $q_\textsf{init}\in Q$ is the initial state,
$\delta\colon Q\times\Sigma\to Q$ is the transition function, and $F\subseteq Q$ is the
set of \emph{safe} states. The \emph{run} of $\mathcal A$ on a word $w$ is the unique
sequence $\rho=q_0q_1\cdots\in Q^\omega$ with $q_0=q_\textsf{init}$ and $q_{i+1}=\delta(q_i,w[i])$
for all $i\ge 0$; the automaton \emph{accepts} $w$ if $q_i\in F$ for all $i\ge 0$. We
write $L(\mathcal A)\subseteq\Sigma^\omega$ for the set of accepted words by $\mathcal A$, and
$|\mathcal A|:=|Q|$ for the size of $\mathcal A$. Every language of this form is
\emph{$\omega$-regular} \cite[Chap.~4]{PMC-book}.
 
\subparagraph*{Markov chains} We follow the standard terminology of~\cite{PMC-book,Puterman94}.
A \emph{Markov chain} (MC) is a tuple $\mathcal C=(S,\Sigma,s_\textsf{init},P,\ell)$, where $S$ is
a countable set of states, $s_\textsf{init}\in S$ is the initial state,
$P\colon S\times S\to[0,1]\cap \Q$ is a transition probability function satisfying
$\sum_{s'\in S}P(s,s')=1$ for every $s\in S$, and $\ell\colon S\times S\to\Sigma$ labels
transitions.
The MC is \emph{finite} if $S$ is finite, in which case we write $|\mathcal C|$ for the size of the
representation of $\mathcal C$: states, actions, labels, and transitions with non-zero probability are encoded explicitly, while probabilities are encoded in binary.
 
A \emph{path} of $\mathcal C$ is an infinite sequence $\pi=s_0s_1\cdots\in S^\omega$ such
that $s_0=s_\textsf{init}$ and $P(s_i,s_{i+1})>0$ for every $i\ge 0$. A \emph{finite path} is a
finite prefix of a path. We write $\Paths(\mathcal C)$ and $\FPaths(\mathcal C)$ for the
sets of paths and of finite paths of $\mathcal C$. The \emph{trace} of a path
$\pi=s_0s_1\cdots$ is the word $\trace(\pi):=\ell(s_0,s_1)\,\ell(s_1,s_2)\cdots$, and the
trace of a finite path $\bar\pi=s_0\cdots s_n$ is the finite word
$\trace(\bar\pi):=\ell(s_0,s_1)\cdots\ell(s_{n-1},s_n)$ of length $n$. For $s\in S$, let $\mathcal C_s$ denote the MC obtained from
$\mathcal C$ by replacing its initial state with $s$. Given a finite path
$\bar\pi=s_0\cdots s_n$ of $\mathcal C$ and a path $\pi$ of $\mathcal C_{s_n}$, we write
$\bar\pi\odot\pi:=s_0\cdots s_{n-1}\cdot\pi$ for their \emph{concatenation}, which is a
path of $\mathcal C$ and satisfies
$\trace(\bar\pi\odot\pi)=\trace(\bar\pi)\cdot\trace(\pi)$.
The \emph{cylinder set} of a finite path $\bar\pi$ is
$\mathrm{Cyl}(\bar\pi):=\{\pi\in\Paths(\mathcal C)\mid\bar\pi\text{ is a prefix of }\pi\}$, and, with overloading notation, for $\bar \pi =s_0\cdots s_n$ denote $P(\bar\pi):=\prod_{i=0}^{n-1}P(s_i,s_{i+1})$, where the empty
product is $1$. 
Every MC $\mathcal C$ induces a probability space
$(\Paths(\mathcal C),\mathcal F_{\mathcal C},\Pr_{\mathcal C})$, where
$\mathcal F_{\mathcal C}$ is the $\sigma$-algebra generated by the cylinder sets and
$\Pr_{\mathcal C}$ is the unique probability measure with
$\Pr_{\mathcal C}(\mathrm{Cyl}(\bar\pi))=P(\bar\pi)$ for every
$\bar\pi\in\FPaths(\mathcal C)$. A language
$L\subseteq\Sigma^\omega$ is \emph{measurable} if it is Borel for the usual topology on
$\Sigma^\omega$; in which case, the event
$E^{\mathcal C}_L:=\{\pi\in\Paths(\mathcal C)\mid\trace(\pi)\in L\}$ belongs to
$\mathcal F_{\mathcal C}$ for every MC $\mathcal C$ over $\Sigma$, and we write
$\Pr_{\mathcal C}(L):=\Pr_{\mathcal C}(E^{\mathcal C}_L)$. Every $\omega$-regular language
is measurable \cite[Chap.~10.3]{PMC-book}.

A set $T\subseteq S$ is \emph{strongly connected} if for all $s,s'\in T$ there are
$t_0,\dots,t_n\in T$ with $t_0=s$, $t_n=s'$ and $P(t_i,t_{i+1})>0$ for all $0\le i<n$. A
\emph{strongly connected component} (SCC) is a maximal strongly connected set, and an SCC
$B$ is a \emph{bottom} SCC (BSCC) if $P(s,s')=0$ for all $s\in B$ and $s'\notin B$. We
write $E_\BSCC(\mathcal C)$ for the event consisting of all paths that enter some BSCC $B$. Following results can be found in~\cite{PMC-book}.
 
\begin{proposition}\label{prop:mc}
Let $\mathcal C$ be a finite MC, let $E\in\mathcal F_{\mathcal C}$, and let $\bar\pi$ be a finite
path of $\mathcal C$ ending in a state $s$. Then:
\begin{enumerate}
  \item $\Pr_{\mathcal C}\bigl(E\mid\mathrm{Cyl}(\bar\pi)\bigr)
        =\Pr_{\mathcal C_s}\bigl(\{\pi\in\Paths(\mathcal C_s)\mid\bar\pi\odot\pi\in E\}\bigr)$;
  \item if $\Pr_{\mathcal C}(E)>0$, then there is a path $\pi\in E$ with
        $\Pr_{\mathcal C}\bigl(E\mid\mathrm{Cyl}(\pi[0..n])\bigr)>0$ for every $n\ge0$;
  \item $\Pr_{\mathcal C}(E_\BSCC(\mathcal C))=1$.
\end{enumerate}
\end{proposition}

\section{The stochastic target discounted-sum problem}
The target discount-sum language (\TDS language for short) is the language of words with discounted-sum equal to the target value. Formally, for a finite alphabet $\Sigma \subset \Z$, a discount factor $\lambda \in (0,1)$ and a target value $t \in \R$, the \emph{\TDS language}  is $L_{[\Sigma,\lambda,t]} := \{w \in \Sigma^\omega \mid \DS(w)=t\}$. Our main contribution is to establish that in finite Markov chains the probability of the \TDS language can be computed in pseudo-polynomial time, thus proving that the stochastic \TDS problem is solvable in pseudo-polynomial time. More precisely:

\begin{theorem}[Resolution of the stochastic \TDS problem] \label{thm:MC} Let $\Sigma \subset \Z$ be a finite alphabet, $\lambda = p/q \in (0,1) \cap \Q$ a discount factor and $t = a/b \in \Q$ a target value, both in reduced form, let $\mu = \max_{\sigma \in \Sigma} |\sigma|$, and let $C$ be a finite Markov chain over $\Sigma$. Then $\Pr_C(L_{[\Sigma,\lambda,t]})$ is rational and computable in time polynomial in $|C| \cdot |\Sigma|\cdot b q^2 \mu$.
\end{theorem}

Without loss of generality, we assume in the sequel computability results that $\Sigma \subset \Z$.
Indeed, any finite rational alphabet can be scaled by a input-dependent constant, together with the target $t$, to obtain an equivalent instance over an integer alphabet.  

To prove the result, we partition the target discounted sum language $L_{[\Sigma,\lambda,t]}$ into two sub-languages $L^{<\infty}_{[\Sigma,\lambda,t]}$ and $L^{\infty}_{[\Sigma,\lambda,t]}$ according to the finiteness of the suffix sums:
\begin{align}
    &L^{<\infty}_{[\Sigma,\lambda,t]}\label{eq:TDS-lang-fin}
:=\left\{
w \in \Sigma^\omega
\;\middle|\;
\DS(w)=t
\text{ and }
\{\DS(w[i..\infty]) \mid i \ge 0\}
\text{ is finite}
\right\}, \notag\\
    &L^{\infty}_{[\Sigma,\lambda,t]}
:=\left\{
w \in \Sigma^\omega
\;\middle|\;
\DS(w)=t
\text{ and }
\{\DS(w[i..\infty]) \mid i \ge 0\}
\text{ is infinite}
\right\}. \notag
\end{align}
One can show that the languages $L_{[\Sigma,\lambda,t]}$, $L^{<\infty}_{[\Sigma,\lambda,t]}$, and $L^{\infty}_{[\Sigma,\lambda,t]}$ are measurable (see Lemma~\ref{lem:measurable} in Appendix~\ref{app:proof}). The proof of~\cref{thm:MC} then consists of two building blocks.
First, we establish a key structural property of finite Markov chains, that implies that $L^{\infty}_{[\Sigma,\lambda,t]}$ has zero measure \footnote{The result even holds more generally for finite alphabet over reals.}.

\begin{theorem}\label{thm:MC-measure-zero}
Let $\Sigma \subset \R$ be a finite alphabet, $\lambda \in (0,1)$ a discount factor, and  $t \in \R$ a target value. For every finite Markov chain $C$ over alphabet $\Sigma$, 
$\Pr_C(L^\infty_{[\Sigma,\lambda,t]}) = 0$.
\end{theorem}

 Then, we study the sub-languages $L^{<\infty}_{[\Sigma,\lambda,t]}$ and  $L^\infty_{[\Sigma,\lambda,t]}$ to obtain a characterization of $\omega$-regularity for the \TDS language $L_{[\Sigma,\lambda,t]}$.

\begin{theorem}\label{thm:omega-regularity}
Let $\Sigma \subset \Z$ be a finite alphabet, $\lambda = p/q \in (0,1) \cap \Q$ a discount
factor and $t = a/b \in \Q$ a target value, with $a/b$ and $p/q$ in reduced form, and let
$\mu = \max_{\sigma \in \Sigma} |\sigma|$. Then:
\begin{enumerate}
  \item The language $L^{<\infty}_{[\Sigma,\lambda,t]}$ is recognized by a deterministic safety automaton $\mathcal{A}_{[\Sigma,\lambda,t]}$ with $O(b q^2 \mu)$ states; moreover, $\mathcal{A}_{[\Sigma,\lambda,t]}$ can be constructed in time $O(b q^2 \mu \cdot |\Sigma|)$.
  \item The language $L^{\infty}_{[\Sigma,\lambda,t]}$ is $\omega$-regular if and only if it is empty.
\end{enumerate}
Thus the target discounted sum language $L_{[\Sigma,\lambda,t]}$  is $\omega$-regular iff $L_{[\Sigma,\lambda,t]} = L^{<\infty}_{[\Sigma,\lambda,t]}.$
\end{theorem}

Assuming that \cref{thm:MC-measure-zero,thm:omega-regularity} hold, we now explain how to complete the proof of \cref{thm:MC}.
By Theorem~\ref{thm:MC-measure-zero}, in a finite Markov chain $C$,
\(\Pr_C(L_{[\Sigma,\lambda,t]}) = \Pr_C(L^{<\infty}_{[\Sigma,\lambda,t]})\) so that one can compute  \(\Pr_C(L^{<\infty}_{[\Sigma,\lambda,t]})\) to solve the stochastic \TDS problem.
Then, by Theorem~\ref{thm:omega-regularity}, $L^{<\infty}_{[\Sigma,\lambda,t]}$ is recognized by deterministic safety automaton $\mathcal A_{[\Sigma,\lambda,t]}$ of size $O(bq^2\mu)$. 
It therefore suffices to compute $\Pr_C(L(\mathcal A_{[\Sigma,\lambda,t]}))$,
which can be done by taking the synchronous product of $C$ with $\mathcal A_{[\Sigma,\lambda,t]}$, and computing there the probability of accepting runs. The latter can be done in time polynomial in $|C| \cdot |\mathcal A_{[\Sigma,\lambda,t]}|$.

The rest of this section is thus devoted to proving \cref{thm:MC-measure-zero,thm:omega-regularity}, followed by some more results for distribution of Markov chains in~\cref{sec:app}.

\subsection{Infinite suffix-sum \TDS languages have measure zero}\label{sec:MC-result}
In this subsection we prove~\cref{thm:MC-measure-zero}. To prove~\cref{thm:MC-measure-zero}, we first prove its statement for the case where $C$ is strongly connected (Lemma~\ref{Lemma:MC-SC-2}). We then use \cref{prop:mc} to generalize this result to arbitrary finite Markov chains. To prove the statement for strongly connected Markov chains, we prove that in a finite, strongly connected Markov chain $C$ over alphabet $\Sigma$, the probability $\Pr_C(L_{[\Sigma, \lambda, t]})$ is either $0$ or $1$ (Lemma~\ref{Lemma:MC-SC-1}).
We begin by introducing the necessary notation.

\paragraph*{Notation and identities} 
Let $\Sigma \subset \R$ be a finite alphabet, $\lambda \in (0,1)$ be a discount factor, and $t \in \R$ be a target value. 
Recall that for an infinite word $w = d_0d_1\dots \in \Sigma^\omega$, its discounted sum is
$\DS(w) = \sum_{i=0}^{\infty} d_i \lambda^i.$
We overload the notation $\DS$ of discounted sums to finite words: for $x = d_0d_1\dots d_n \in \Sigma^*$, $\DS(x) = \sum_{i=0}^{n} d_i \lambda^i.$
We emphasize the following recursive property: for every word $w \in \Sigma^\omega$ and every index $n > 0$
\begin{equation}\label{eq:DS-recursive}
    \DS(w) = \DS(w[0\dots n-1]) + \lambda^n\DS(w[n\dots \infty]).
\end{equation}
Let $C = (S, \Sigma, s_{\mathsf{init}}, P, \ell)$ be a finite Markov chain. When $C$ is clear from context, for $s$ a state of $C$, we write $\Pr_s$ as a shorthand for $\Pr_{C_s}$, the probability measure in the Markov chain with initial state $s$. 
Recall that if $L \subseteq \Sigma^\omega$ is a language, the notation $\Pr_s(L)$ represents the probability of the event that a random infinite path in $C_s$ has a trace in $L$. 
Therefore, conditional probabilities such as $\Pr_s(L \mid E)$ or intersections like $\Pr_s(L \cap E)$ denote refer to operations on the corresponding sets of paths in $C_s$, rather than on the language $L$ itself.
For a finite or infinite path $\pi$, we write $\DS(\pi)$ as a shorthand for $\DS(\textit{trace}(\pi))$. Let $\hat \pi = s_0 s_1 \dots s_n$ be a finite path in $C_s$. We also introduce the shorthands $L_x$ and $L_x^\infty$ for the languages $L_{[\Sigma,\lambda,x]}$ and $L_{[\Sigma,\lambda,x]}^\infty$, respectively, whenever $\Sigma$ and $\lambda$ are clear from context.
By \cref{prop:mc} and Equation~\eqref{eq:DS-recursive}, the language $L_t$ satisfies
\begin{equation}\label{eq:claim-MC}
    \Pr_{s}(L_t \mid \mathrm{cyl}(\hat \pi)) = \Pr_{s_n}(L_{t_{\hat \pi}}), \quad \text{where} \quad  t_{\hat \pi} = (t - \DS({\hat \pi}))/\lambda^n.
\end{equation}

Analogously, consider the language $L_t^\infty$. The property of having infinitely many distinct suffix sums is a tail property. In other words, for every index $n$, an infinite word $w \in \Sigma^\omega$ has infinitely many distinct suffix sums if and only if its suffix $w[n \dots \infty]$ also does. Therefore
\begin{equation}\label{eq:claim-MC-infity}
    \Pr_{s}(L_t^\infty \mid \mathrm{cyl}(\hat{\pi})) = \Pr_{s_n}(L_{t_{\hat \pi}}^\infty), \quad \text{where} \quad  t_{\hat \pi} = (t - \DS({\hat \pi}))/\lambda^n.
\end{equation}

\paragraph*{Strongly connected Markov chains}
We now state a strong version of a 0-1 law for $L_{[\Sigma, \lambda, t]}$ in strongly connected Markov chains.
\begin{lemma}\label{Lemma:MC-SC-1}Let $\Sigma \subset \R$ be a finite alphabet, $\lambda \in (0,1)$ a discount factor, and  $t \in \R$ a target value. For every finite, strongly connected Markov chain $C$ over alphabet $\Sigma$, either $\Pr_C(L_{[\Sigma, \lambda, t]}) =0 $, or every infinite path $\pi$ in $C$ satisfies $\DS(\pi) = t$.

\end{lemma} 
\begin{proof}
Let $C = (S, \Sigma, s_{\mathsf{init}}, P, \ell)$ be a finite, strongly connected Markov chain.
For a state $s \in S$, define the supremum
$p^*_s = \sup_{x \in \R} \Pr_{s}(\lang{x}).$
We first prove that supremum $p_s^*$ is attained; that is, there exists a $v_{s} \in \R$ such that $\Pr_{s}(\lang{v_s}) = p^*_s$: If $p^*_s = 0$, then any target real number attains the supremum $p_s^*$. If $p^*_s > 0$, then there exists a real number $x_0$ such that $\Pr_{s}(\lang{x_0}) > 0$. Because the non-empty languages $L_x$ and $L_y$ are disjoint for every $x \neq y \in \R$, the sum of probabilities $\Pr_{s}(\lang{x})$ over all $x\in \R$ sum to at most $1$. Thus, at most $\frac{1}{\Pr_{s}(\lang{x_0})}$ many real numbers $x$ satisfy $\Pr_{s}(L_x) \ge \Pr_{s}(\lang{x_0})$. Consequently, the supremum $p_s^*$ is a maximum taken over a non empty finite set, and thus it is attained. Let $v_s \in \R$ be such that 
\begin{equation}\label{eq:max_attained}
    \Pr_{C_{s}}(\lang{v_s}) = p^*_s.
\end{equation}
Define $p^* = \max_{s \in S}p^*_s.$ We next show that for every state $s \in S$, $p^*_s = p^*$. To do so, we establish the following claim:
\begin{claim}\label{claim:p^*-1}
    For every state $s \in S$ and target number $x \in \R$ such that $\Pr_{s}(\lang{x}) = p^*$, for every finite path $\hat \pi = s_0s_1 \dots s_n$ in $C_s$ of length $n+1 \geq 0$,
    $\Pr_{{s_n}}(\lang{x_{\hat \pi}}) = p^*,$
    where $x_{\hat \pi} = \frac{x - \DS(\hat \pi)}{\lambda^n}.$
\end{claim}
\begin{proof}[Proof of claim] Fix an arbitrary path length $n+1 \geq 0$. The law of total probability gives:
\(
    \Pr_{s}(\lang{x}) =
    \sum_{\hat \pi =s_0s_1\dots s_n \in  \FPaths(C_s)} \Pr_{s}( \mathrm{cyl}(\hat \pi)) \cdot \Pr_{s}(\lang{x} \mid \mathrm{cyl}(\hat \pi)).
\)
Note that the sum is taken over only finite paths of length $n+1$.
Substituting Equation~\eqref{eq:claim-MC}, we obtain:
\(
    \Pr_{s}(\lang{x}) =
    \sum_{\hat \pi =s_0s_1\dots s_n \in  \FPaths(C_s)} \Pr_{s}( \mathrm{cyl}(\hat \pi)) \cdot \Pr_{{s_n}}(\lang{x_{\hat \pi}}).
\)
where $x_{\hat \pi}$ is as defined in \cref{claim:p^*-1}.

By definition of $p^*$, for every finite path $\hat \pi$ of length $n+1$,  $\Pr_{{s_n}}(\lang{x_{\hat \pi}}) \leq p^*$. Because the cylinder probabilities $\Pr_{s}( \mathrm{cyl}(\hat \pi))$ are positive and sum to $1$, the right-hand side of the above equation is a convex combination of terms bounded by $p^*$. Consequently, the assumption $\Pr_{s}(\lang{x}) = p^*$ implies that every term $\Pr_{{s_n}}(\lang{x_{\hat \pi}})$ is equal to $p^*$.
\end{proof}
\begin{claim}\label{claim:p^*-2}
    For every state $s \in S$, $p^*_s = p^*$.
\end{claim}
\begin{proof}[Proof of claim]
    Recall that $p^* = \max_{s \in S}p^*_s$. Let $u \in S$ be a state such that $p_{u}^* = p^*$. By~\eqref{eq:max_attained}, $\Pr_{{u}}(\lang{v_u}) = p^*$. Because the Markov chain $C$ is strongly connected, for every state $s\in S$ there is a finite path in $C_{u}$ of positive length ending at $s$. By the preceding claim, this path guarantees the existence of a target number $x_0 \in \R$ such that $\Pr_{s}(\lang{x_0}) = p^*$.
\end{proof}

We now prove the lemma. Let $t \in \R$ be such that $\Pr_C(L_{[\Sigma, \lambda, t]}) >0$, or simply
\begin{equation}
    \label{eq:lem4-assm}
    \Pr_{s_{\textsf{init}}}(\lang{t}) > 0.
\end{equation}
We wish to show that every infinite path $\pi$ in $C$ satisfies $\DS(\pi) = t$. To do this we will show that every infinite path $\pi$ in $C$ satisfies $\DS(\pi)=v_{s_{\mathsf{init}}}$. Since $\Pr_{s_{\textsf{init}}}(\lang{t}) > 0$, this implies $t = v_{s_{\mathsf{init}}}$, which would conclude the proof. 

Assume for contradiction that there exists an infinite path $\pi_0 = s_0s_1\dots$ in $C$ such that  $v' = \DS(\pi_0) \neq v_{s_{\mathsf{init}}}$. 
For index $n \ge 0$, let $y_n := \DS(\pi_0[n\dots\infty])$. The recursive property of discounted sums (Equation~\eqref{eq:DS-recursive}) gives:
$y_n = \frac{v' - \DS(\pi_0[0\dots n])}{\lambda^n}.$
By Claim~\ref{claim:p^*-2} and Equation~\eqref{eq:max_attained} $\Pr_{s_{\textsf{init}}}(\lang{v_{s_{\mathsf{init}}}}) = p^*$, thus for the finite path $\pi_0[0\dots n]$ of $C_{s_{\textsf{init}}}$, Claim~\ref{claim:p^*-1} implies that for
$x_n := \frac{v_{s_{\mathsf{init}}} - \DS(\pi_0[0\dots n])}{\lambda^n},$
one has $\Pr_{{s_n}}(\lang{x_n}) = p^*$. By Equation~\eqref{eq:lem4-assm}, $p^* > 0$. Hence, $\Pr_{{s_n}}(\lang{x_n})>0$  and thus there exists an infinite path $\pi_n$ in $C_{s_n}$ such that $\DS(\pi_n) = x_n$. We derive a contradiction for sufficiently large $n$. 
Let $d_{\min} = \min_{d \in \Sigma} d$ and $d_{\max} = \max_{d\in \Sigma} d$. The absolute difference between the discounted sums of every two infinite paths in a Markov chain over alphabet $\Sigma$ is bounded by:
$W = \sum_{i=0}^{\infty} \lambda^i (d_{\max} - d_{\min}) = \frac{d_{\max} - d_{\min}}{1 - \lambda}.$
Thus for every $n\geq 0$, the definitions of $\pi_n$, $x_n$ and $y_n$ guarantee
\begin{equation}\label{eq:ab-val-MC}
    |\DS(\pi_n) -\DS(\pi_0[n\dots\infty])| = |x_n -y_n| = \frac{|v_{s_{\mathsf{init}}} - v'|}{\lambda^n}\leq W
\end{equation}
Because $0 <\lambda < 1$ and $v_{s_{\mathsf{init}}} \neq v'$, the sequence $\frac{|v_{s_{\mathsf{init}}} - v'|}{\lambda^n}$ diverges as $n \to \infty$. Thus, there exists a sufficiently large positive integer $n$ such that
\(\frac{|v_{s_{\mathsf{init}}} - v'|}{\lambda^n}> W,\)
contradicting Equation~\eqref{eq:ab-val-MC}. Therefore, our assumption that $\pi_0$ exists is false, proving the lemma. 
\end{proof}

\begin{lemma}\label{Lemma:MC-SC-2}
Let $\Sigma \subset \R$ be a finite alphabet, $\lambda \in (0,1)$ a discount factor, and  $t \in \R$ a target value. For every finite, strongly connected Markov chain $C$ over alphabet $\Sigma$, $\Pr_C(L^\infty_{[\Sigma,\lambda,t]}) = 0.$
\end{lemma}
\begin{proof}
Let $C = (S, \Sigma, s_{\mathsf{init}}, P, \ell)$ be a finite, strongly connected Markov chain. Assume for contradiction that $\Pr_{s_{\textsf{init}}}(\langinf{t}) > 0$.  
By \cref{prop:mc}, there exists a path $\pi= s_0s_1\ldots$ in $C$ such that $\trace(\pi) \in \langinf{t}$ and for every $n \ge 0$, the conditional probability satisfies:
$\Pr_C(\langinf{t} \mid \mathrm{Cyl}(\pi[0\dots n])) > 0.$
Substituting Equation~\eqref{eq:claim-MC-infity} gives:
\(
    \Pr_{{s_n}}(\langinf{t_n}) > 0, \quad \text{where } t_{n} = \frac{t - \DS(\pi[0\dots n])}{\lambda^n}.
\)
Because $\langinf{t_n} \subseteq \lang{t_n}$, it follows that for every $n \geq 0$:
\begin{equation}\label{eq:>0lemma-MC-2} \Pr_{{s_n}}(\lang{t_n}) > 0.\end{equation}
Because $C$ is strongly connected, the Markov chain $C_{s_n}$ is also strongly connected. Thus, by Lemma~\ref{Lemma:MC-SC-1}, every path in $C_{s_n}$ has discounted sum $t_n$.

Since $\trace(\pi) \in \langinf{t}$, by the recursive property of discounted sums (Equation~\eqref{eq:DS-recursive}) observe that $t_{n} = \DS(\pi[n \dots\infty])$ and hence $\{t_n \mid n \geq 0\}$ is infinite.
Because the sequence of states $s_n$ is drawn from the finite set $S$, while the set of values $\{t_n \mid n \geq 0\}$ is infinite, the pigeonhole principle dictates that there exist indices $k ,m \geq 0$ with $k \neq m$ such that $s_k = s_m$ but $t_k \neq t_m$. 
However, $s_k = s_m$ implies that the Markov chains $C_{s_k}$ and $C_{s_m}$ are identical. Moreover, because $C$ is strongly connected, $C_{s_k}$ and $C_{s_m}$ are also strongly connected. As established in Equation~\eqref{eq:>0lemma-MC-2} both $\Pr_{{s_k}}(\lang{t_k}) > 0$ and $\Pr_{{s_m}}(\lang{t_m}) > 0$. By Lemma~\ref{Lemma:MC-SC-1}, all paths in this Markov chain must evaluate to a single, unique discounted sum, which forces $t_k =t_m$. This contradicts the choice of $k$ and $m$, completing the proof.
\end{proof}

\begin{proof}[Proof of Theorem~\ref{thm:MC-measure-zero}]  Using \cref{Lemma:MC-SC-2} we now prove \cref{thm:MC-measure-zero}.
Let $C= (S, \Sigma, s_{\mathsf{init}}, P, \ell)$ be a finite Markov chain. Assume for contradiction that $\Pr_{s_{\textsf{init}}}(\langinf{t}) > 0$.   
By \cref{prop:mc}, $\Pr_C(E_{\BSCC}(C)) = 1$. Thus, \(\Pr_{s_{\textsf{init}}}\big(\langinf{t} \cap E_{\BSCC}(C)\big) = \Pr_{s_{\textsf{init}}}\big(\langinf{t}\big)> 0.\)
By \cref{prop:mc}, there exists a path $\pi = s_0s_1\ldots \in E_{\BSCC}(C)$ such that $\trace(\pi) \in \langinf{t}$ and for every $n \ge 0$: $\Pr_{s_{\textsf{init}}}\big(\langinf{t} \cap E_{\BSCC}(C) \mid \mathrm{Cyl}(\pi[0\dots n])\big) > 0.$
Therefore $\Pr_{s_{\textsf{init}}}(\langinf{t} \mid \mathrm{Cyl}(\pi[0\dots n])) > 0.$ Substituting Equation~\eqref{eq:claim-MC-infity} gives:
\begin{equation}\label{eq:thm-MC-1}
    \Pr_{{s_n}}(\langinf{t_n}) > 0, \quad \text{where } t_{n} = \frac{t - \DS(\pi[0\dots n])}{\lambda^n}.
\end{equation}
Because $\pi \in E_{\BSCC}(C)$, the path eventually reaches a BSCC $B$, that is there exists an index $m > 0$ such that $s_m \in B$. 
However, this results in a contradiction. The Markov chain $C_{s_m}$ induces the same probability space as the finite, strongly connected Markov chain $C' = (B, \Sigma, s_m, P', \ell)$, where $P'$ is the projection of $P$ onto the domain $B \times B$. Thus, Lemma~\ref{Lemma:MC-SC-2} guarantees that $\Pr_{C_{s_m}}(\langinf{t_m})  = \Pr_{C'}(\langinf{t_m}) = 0$. This contradicts Equation~\eqref{eq:thm-MC-1} for index $m$, completing the proof.
\end{proof}

\subsection{Omega-regularity characterization of \TDS languages}\label{sec:Automata}

We now prove Theorem~\ref{thm:omega-regularity} that provides a characterization of $\omega$-regularity of \TDS languages. To do this, we first present a construction of a deterministic safety automaton $\mathcal{A}_{[\Sigma,\lambda,t]}$ (Definition~\ref{def:automata}) and show in Lemma~\ref{lemma:Automata-correctness} that when $\Sigma \subset \Z$ is finite, $\lambda \in (0,1) \cap \Q$, and $t \in \Q$, the language of $\mathcal{A}_{[\Sigma,\lambda,t]}$ is equal to the language $L^{<\infty}_{[\Sigma,\lambda,t]}$. 

Let $\Sigma \subset \Z$ be a finite alphabet, $\lambda \in (0,1) \cap\Q$ a discount factor, and $t \in \Q$ a target value. Intuitively, the states of the automaton $\mathcal{A}_{[\Sigma,\lambda,t]}$ store, after reading a finite word $d_0d_1\dots d_n$, the current \emph{required suffix-sum}, i.e., the discounted sum that the remaining suffix $d_{n+1}d_{n+2}\dots$ must have for the discounted sum of the entire word $d_0d_1d_2\dots$ to be equal to $t$. Formally, a state stores a unique value $q \in \Q$ that satisfies
\(
t
=
\sum_{i=0}^{n} d_i \lambda^i
+
\lambda^{n+1} q.
\)
Thus locally, if the current required suffix sum is $q$, then after reading a symbol $d \in \Sigma$, the remaining suffix must have discounted sum $q'$ where
$q = d + \lambda q'$. For example, if $\lambda=\frac12$ and $t=2$, then after reading the prefix $d_0=1, d_1=2$, the automaton stores the value $0$. 
Using required suffix-sums as automaton states has appeared before~\cite{BCV22,bertrand2026categorizerautomatadiscountedsumpayoffs}.
However, when the discount factor $\lambda$ is a rational that is not an inverse integer, the set of required suffix-sums that an automaton may need to track is potentially infinite, preventing a finite-state construction~\cite{BCV22}. The finiteness condition in the definition of $L^{<\infty}_{[\Sigma,\lambda,t]}$ allows one to resolve this difficulty.  

The following lemma establishes that only finitely many distinct suffix-sums are required after reading prefixes of words in $L^{<\infty}_{[\Sigma,\lambda,t]}$. The proof of the lemma is provided in Appendix~\ref{app:proof} and is an adaptation of the denominator analysis of ``gaps'' in Boker et al.~\cite{boker_target_2015}. 
However, where Boker et al. utilize this analysis primarily as an intermediate step towards the decidability of the \TDS problem in special cases, we identify the language of finite suffix-sums and give a deterministic safety automaton recognizing it.

\begin{lemma}[Finiteness of Required Suffix-Sums]\label{lemma:Q}
Let $\Sigma \subset \Z$ be a finite alphabet, $\lambda \in (0,1) \cap\Q$ a discount factor, and $t \in \Q$ a target value.
There exists a finite set $Q_{[\Sigma,\lambda,t]} \subset \Q$ such that for every word $w = d_0 d_1 \dots \in L^{<\infty}_{[\Sigma,\lambda,t]}$ and for every index $i \ge 0$:
\(
\DS(w[i\dots\infty]) \in Q_{[\Sigma,\lambda,t]}.
\)
Specifically, let $\lambda = \frac{p}{q}$ and $t = \frac{a}{b}$, in reduced form, where $p,q,b\in \mathbb{Z}^+$ and $a \in \mathbb{Z}$. 
Defining $d_{\min} = \min_\Sigma \sigma$ and $d_{\max} = \max_\Sigma \sigma$, the set $Q_{[\Sigma,\lambda,t]}$ is defined as
\[
Q_{[\Sigma,\lambda,t]} := \left\{ \frac{x}{bp} \;\middle|\; x \in \mathbb{Z} \text{ and } \frac{d_{\min}}{1-\lambda} \le \frac{x}{bp} \le \frac{d_{\max}}{1-\lambda} \right\}.
\]
\end{lemma}

We now formalize the construction of the deterministic safety automaton $\mathcal{A}_{[\Sigma,\lambda,t]}$. Its state space is $Q_{[\Sigma,\lambda,t]} \cup \{\bot\}$, with the intuition that a state in $Q_{[\Sigma,\lambda,t]}$ represents the required suffix sum, whereas the state $\bot$ indicates that the required suffix sum does not belong to $Q_{[\Sigma,\lambda,t]}$. In the latter case, \cref{lemma:Q} guarantees that no suffix may result in a word in $L^{<\infty}_{[\Sigma,\lambda,t]}$, thus making $\bot$ an unsafe state for $\mathcal{A}_{[\Sigma,\lambda,t]}$.

\begin{definition}[Deterministic Safety Automaton $\mathcal{A}_{[\Sigma,\lambda,t]}$]\label{def:automata}
$\mathcal{A}_{[\Sigma,\lambda,t]} := (Q, \Sigma,  q_0, \delta, F)$ is the deterministic safety automaton defined as follows. $Q = Q_{[\Sigma,\lambda,t]} \cup \{\bot\}$ is the set of states;  $q_0 = t$ is the initial state (if $t \notin Q_{[\Sigma,\lambda,t]}$ then $q_0 = \bot$); $F = Q_{[\Sigma,\lambda,t]}$ is the set of safe states; and  $\delta : Q \times \Sigma \to Q$ is the transition function defined by: for $d \in \Sigma$, $\delta(\bot, d) = \bot$ and for all $z \in Q_{[\Sigma,\lambda,t]}$: $\delta(z, d)$ is equal to $\frac{z - d}{\lambda}$ if $\frac{z - d}{\lambda} \in Q_{[\Sigma,\lambda,t]}$ and $\bot$ otherwise.

\end{definition}
The correctness of this construction, formalized below, is proved in \cref{app:proof}.
\begin{lemma}[Correctness]\label{lemma:Automata-correctness}
Let $\Sigma \subset \Z$ be a finite alphabet, $\lambda \in (0,1) \cap\Q$ a discount factor, and $t \in \Q$ a target value. Then
\(
L(\mathcal{A}_{[\Sigma,\lambda,t]}) = L^{<\infty}_{[\Sigma,\lambda,t]}.
\)
\end{lemma}

\begin{proof}[Proof of Theorem~\ref{thm:omega-regularity}]
    The $\omega$-regularity of language $L^{<\infty}_{[\Sigma,\lambda,t]}$ follows from Lemma~\ref{lemma:Automata-correctness}. 
     In \cref{app:proof} we establish that the deterministic safety automaton $\mathcal{A}_{[\Sigma,\lambda,t]}$ has a pseudo-polynomial size.
    We here show that $L^\infty_{[\Sigma,\lambda,t]}$ is not $\omega$-regular unless it is empty. 
    Assume for the sake of contradiction that $L^\infty_{[\Sigma,\lambda,t]}$ is non-empty and $\omega$-regular. A fundamental property of $\omega$-regular languages is that every non-empty $\omega$-regular language contain at least one ultimately periodic word~\cite{DBLP:books/daglib/0016866}. Thus, there exists a word $w \in L^\infty_{[\Sigma,\lambda,t]}$ of the form $w = u v^\omega$, where $u \in \Sigma^*$ and $v \in \Sigma^+$. Since an ultimately periodic word has finitely many distinct suffixes, the set of suffix sums of $w$ is finite, which contradicts the fact that $w \in L^\infty_{[\Sigma,\lambda,t]}$. Finally since $L^\infty_{[\Sigma,\lambda,t]} = L_{[\Sigma,\lambda,t]} \setminus L^{<\infty}_{[\Sigma,\lambda,t]}$, by the fact that $\omega$-regular languages are closed under set difference we get $L_{[\Sigma,\lambda,t]}$ is $\omega$-regular iff $L_{[\Sigma,\lambda,t]} = L^{<\infty}_{[\Sigma,\lambda,t]}$.
\end{proof}

\subsection{Finer analysis of the discounted-sum distribution in Markov chains}\label{sec:app}
Beyond answering the stochastic \TDS problem, our techniques permit a finer analysis of the distribution of discounted-sum payoffs in finite Markov chains. Let $\Sigma\subset\Z$ be finite, $\lambda=p/q\in(0,1)\cap\Q$
in reduced form, $\mu=\max_{\sigma\in\Sigma}|\sigma|$, and
$C$ a finite Markov chain over $\Sigma$.
\begin{theorem}\label{thm:atoms}
It is decidable in time polynomial in $|C| \cdot \log q \cdot \log(1+\mu)$ whether there exists $t\in\R$ such that $\Pr_C(L_{[\Sigma,\lambda,t]})>0$.
\end{theorem}
\begin{proof}
By \cref{prop:mc}(1),(3), and Lemma~\ref{Lemma:MC-SC-1}, there exists $t\in\R$ such that
$\Pr_C(L_{[\Sigma,\lambda,t]})>0$ iff some reachable BSCC contains a state
from which the minimum and maximum discounted sums of paths coincide.
For each reachable BSCC $B$, regard each positive transition
as a deterministic action, retaining its label as the reward.
The two extrema are then the minimum and maximum expected
discounted rewards in a finite MDP and can be computed exactly
by linear programming in $|B|$, $\log q$, and
$\log(1+\mu)$~\cite[Sec.~6.9]{Puterman94}.
The reachable BSCCs can be identified in time polynomial
in $|C|$~\cite{PMC-book}, giving the claimed bound.
\end{proof}

\begin{theorem}\label{thm:MC-threshold}
Given $t\in\Q$, $\varepsilon\in\Q_{>0}$ and $\sim\in\{>,\ge,<,\le\}$, one can compute the probability $\Pr_C(\{w\in\Sigma^\omega\mid\DS(w)\sim t\})$ up to precision $\varepsilon$, i.e., one can compute a rational $p$ such that $0 \le \Pr_C(\{w\in\Sigma^\omega\mid\DS(w)\sim t\}) - p \le \varepsilon$. 
\end{theorem}
\begin{proof}
We give the algorithm for $>t$; its correctness and termination
are proved in Appendix~\ref{app:proof}.
Compute $e=\Pr_C(L_{[\Sigma,\lambda,t]})$ using
Theorem~\ref{thm:MC}.
For $n=0,1,\ldots$, let $r_n=\mu\lambda^n/(1-\lambda)$ (the absolute bound on the contribution to the discounted sum of all transitions beyond the first $n$)
and enumerate the finite paths
$\bar\pi=s_0\cdots s_n\in\FPaths(C)$.
Let $A_n$ be the sum of $P(\bar\pi)$ over those satisfying
$\DS(\bar\pi)>t+r_n$, and $B_n$ the corresponding sum
for $\DS(\bar\pi)<t-r_n$.
Stop when $(1-e -B_n) -A_n\le\varepsilon$ and return $A_n$.
Adding $e$ gives an approximation for $\ge t$;
taking complements gives the approximations for $\le t$ and $<t$.
\end{proof}

\section{Target discounted-sum objectives in stochastic games}
\label{sec:applications}

Throughout this section, we explore when the arguments for solving the stochastic \TDS problem lifts to richer models of Markov decision processes and 2-player stochastic games. In particular the guiding question is: under what restrictions on the players' strategies can the \TDS objective $L_{[\Sigma,\lambda,t]}$ be replaced, without altering the optimal value and strategies, by its $\omega$-regular sub-objective $L^{<\infty}_{[\Sigma,\lambda,t]}$?
We show that, when the maximizing player (if they exist) is restricted to finite memory strategies --while the minimizing player (if they exist) may use arbitrary strategies-- the two objectives are interchangeable.
Since $L^{<\infty}_{[\Sigma,\lambda,t]}$ is recognized by a deterministic safety automaton of pseudo-polynomial size (by Theorem~\ref{thm:omega-regularity}), this equivalence allows every value and synthesis problem to be reduced to the corresponding problems for stochastic games with standard safety objectives~\cite{Fijalkow_2026}.
We formalize and prove this first for MDPs, also known as $1$-player stochastic games (Theorem~\ref{thm:MDP-dec}), and then for general stochastic games (Theorem~\ref{thm:SG-dec}), of which deterministic $2$-player games are a special case.

\subsection{\TDS objectives in Markov decision processes}

\subparagraph*{Preliminaries}
An MDP is a tuple $\mathcal M=(S,\Sigma,s_{\mathsf{init}},\mathit{Act},P,\ell)$,
where $S$ is a finite set of states, $\Sigma$ is a finite alphabet,
$s_{\mathsf{init}}\in S$ is the initial state, $\mathit{Act}$ is a finite set of
actions, $P\colon S\times\mathit{Act}\times S\to[0,1]\cap\Q$ is a transition
function, and $\ell\colon S\times\mathit{Act}\to\Sigma$ is a labeling function.
For every state $s$ and action $a$ we require $\sum_{s'\in S}P(s,a,s')\in\{0,1\}$;
the action $a$ is \emph{enabled} at $s$ if this sum equals $1$, and we assume that
every state has at least one enabled action.
A \emph{history} is a sequence $h=s_0a_0s_1\cdots a_{n-1}s_n$ with
$s_0=s_{\mathsf{init}}$ and $P(s_i,a_i,s_{i+1})>0$ for all $0\leq i<n$; we write
$\mathrm H_{\mathcal M}$ for the set of histories.
Writing $\mathrm{Dist}(\mathit{Act})$ for the set of probability distributions over
$\mathit{Act}$, a \emph{strategy} is a function
$\alpha\colon\mathrm H_{\mathcal M}\to\mathrm{Dist}(\mathit{Act})$ with
$\alpha(h)(a)=0$ whenever $a$ is not enabled at the last state of $h$. Denote
$\Theta_{\mathcal M}$ as the set of all strategies.
A strategy has \emph{finite memory} if it can be implemented by a finite-state machine
$(Q,q_{\mathsf{init}},\delta,\tau)$, where $Q$ is a finite set of memory states,
$q_{\mathsf{init}}\in Q$ is the initial memory state,
$\delta\colon Q\times\mathit{Act}\times S\to Q$ is the memory-update function, and
$\tau\colon Q\times S\to\mathrm{Dist}(\mathit{Act})$ selects a distribution over the
enabled actions; such a strategy uses memory of size $|Q|$. We write
$\Theta^{\FM}_{\mathcal M}$ for the set of finite-memory strategies \cite[Def.~10.97]{PMC-book}. We write $|\mathcal M|$ for the size of the
representation of $\mathcal M$: states, actions, labels, and transitions with non-zero probability are encoded explicitly, while transition probabilities are encoded in binary.

Under a strategy $\alpha$, the MDP $\mathcal M$ induces a (possibly countably
infinite) MC $\mathcal C^{\mathcal M,\alpha}$ over $\Sigma$~\cite[Def.~10.92]{PMC-book}.
If $\alpha$ is finite memory, then $\mathcal C^{\mathcal M,\alpha}$ is finite.
For a measurable language $L\subseteq\Sigma^\omega$ denote
$V_{\mathcal M}^{\alpha}(L):=\Pr_{\mathcal C^{\mathcal M,\alpha}}(L)$, and for
$\opt\in\{\sup,\inf\}$ define
\(
  V_{\mathcal M}^{\opt}(L):=\opt_{\alpha\in\Theta_{\mathcal M}}V_{\mathcal M}^{\alpha}(L),\) and \(
  V_{\mathcal M}^{\FM\text{-}\opt}(L):=\opt_{\alpha\in\Theta^\FM_{\mathcal M}}V_{\mathcal M}^{\alpha}(L).
\)
A strategy $\alpha\in\Theta_{\mathcal M}$ is \emph{$\opt$-optimal} for $L$ if
$V_{\mathcal M}^{\alpha}(L)=V_{\mathcal M}^{\opt}(L)$, and
$\alpha\in\Theta^\FM_{\mathcal M}$ is \emph{$\FM\text{-}\opt$-optimal} for $L$ if
$V_{\mathcal M}^{\alpha}(L)=V_{\mathcal M}^{\FM\text{-}\opt}(L)$.

\paragraph*{Optimal values and strategies}\label{sec:MDP-result}
Let $\Sigma\subset\Z$ be a finite alphabet, $\lambda = p/q\in(0,1)\cap\Q$ a discount factor, $t=a/b\in\Q$ a target value, with $a/b$ and $p/q$ in reduced form. Let $\mathcal M$ be a finite MDP over alphabet $\Sigma$. Let $\mu=\max_{\sigma\in\Sigma}|\sigma|$.
We first prove \cref{lem:MDP-inf}, which relates the $\inf$- and $\sup$-optimal values of the objectives $L_{[\Sigma,\lambda,t]}$ and $L^{<\infty}_{[\Sigma,\lambda,t]}$.

\begin{lemma}\label{lem:MDP-inf}
\(V^{\inf}_{\mathcal M}(L_{[\Sigma,\lambda,t]}) = V^{\inf}_{\mathcal M}(L^{<\infty}_{[\Sigma,\lambda,t]}) \quad \text{ and }\quad V^{\textsf{FM-}\sup}_{\mathcal M}(L_{[\Sigma,\lambda,t]}) = V^{\sup}_{\mathcal M}(L^{<\infty}_{[\Sigma,\lambda,t]}).\)

\end{lemma}
\begin{proof}
Let $\opt \in \{\sup,\inf\}$. Since every finite-memory strategy $\alpha$ induces a finite Markov chain, Theorem~\ref{thm:MC-measure-zero} gives
$V^\alpha_{\mathcal M}(L_{[\Sigma,\lambda,t]}) =
V^\alpha_{\mathcal M}(L^{<\infty}_{[\Sigma,\lambda,t]})$
for every $\alpha \in \Theta^\FM_{\mathcal M}$. Hence
\begin{equation}
    V^{\FM-\opt}_{\mathcal M}(L_{[\Sigma,\lambda,t]})
=
V^{\FM-\opt}_{\mathcal M}(L^{<\infty}_{[\Sigma,\lambda,t]})
=
V^{\opt}_{\mathcal M}(L^{<\infty}_{[\Sigma,\lambda,t]}).
\label{eq:fm-equality}
\end{equation}
where the second equality follows from Theorem~\ref{thm:omega-regularity} and the standard result that finite-memory strategies suffice for $\omega$-regular objectives in finite MDPs. For $\opt=\sup$, this gives the desired equality. For $\opt=\inf$, since
$L^{<\infty}_{[\Sigma,\lambda,t]} \subseteq L_{[\Sigma,\lambda,t]}$,
we have
$V^{\inf}_{\mathcal M}(L^{<\infty}_{[\Sigma,\lambda,t]})
\leq V^{\inf}_{\mathcal M}(L_{[\Sigma,\lambda,t]})$.
Moreover, since $\Theta^\FM_{\mathcal M}\subseteq\Theta_{\mathcal M}$, we have 
$V^{\inf}_{\mathcal M}(L_{[\Sigma,\lambda,t]})
\leq V^{\FM-\inf}_{\mathcal M}(L_{[\Sigma,\lambda,t]})$.
Together with \eqref{eq:fm-equality}, these inequalities yield
$V^{\inf}_{\mathcal M}(L_{[\Sigma,\lambda,t]})
=
V^{\inf}_{\mathcal M}(L^{<\infty}_{[\Sigma,\lambda,t]})$.
\end{proof}

The following results follow by combining Lemma~\ref{lem:MDP-inf}, Theorem~\ref{thm:omega-regularity}, and standard results for safety objectives in finite MDPs~\cite{PMC-book}. The language $L^{<\infty}_{[\Sigma,\lambda,t]}$ admits a deterministic safety automaton of pseudo-polynomial size, allowing its optimal value to be computed and a deterministic finite-memory optimal strategy to be synthesized in pseudo-polynomial time. Lemma~\ref{lem:MDP-inf} then relates these values and strategies to the original target discounted-sum objective in the MDP. See \cref{app:proof} for complete proofs.

\begin{theorem}\label{thm:MDP-dec}
For the \TDS objective $L_{[\Sigma,\lambda,t]}$:
\begin{enumerate}
    \item the $\inf$-optimal value is rational, can be computed, and an $\inf$-optimal deterministic strategy using $O(b q^2 \mu)$ memory can be synthesized, both in time polynomial in $|{\mathcal M}| \cdot |\Sigma|\cdot b q^2 \mu$;
    \item the $\FM$-$\sup$-optimal value is rational, can be computed, and an $\FM$-$\sup$-optimal deterministic strategy using $O(b q^2 \mu)$ memory can be synthesized, both in  time polynomial in $|{\mathcal M}| \cdot |\Sigma|\cdot b q^2 \mu$.
\end{enumerate}
\end{theorem}

The second item in~\cref{thm:MDP-dec} yields a dichotomy for $\sup$-optimal strategies (when they exist): either pseudo-polynomially bounded memory suffices, or infinite memory is necessary.
We provide evidence that the latter case does occur, by adapting the proof of~\cite[Proposition~1]{chatterjee_multi-objective_2013} to exhibit a finite MDP in which a $\sup$-optimal strategy exists, yet every such strategy needs infinite memory.
 \begin{proposition}
    For $\Sigma = \{0, 1\}$,  $\lambda = \frac23$, and $t = \frac32$, there exists an MDP ${\mathcal M}$ over alphabet $\Sigma$ such that for the \TDS objective $L_{[\Sigma, \lambda, t]}$, a $\sup$-optimal strategy $\alpha^* \in \Theta_{\mathcal M}$ exists and satisfies
    \(V^{\alpha^*}_{\mathcal M}(L_{[\Sigma,\lambda,t]}) > V^{\FM-\sup}_{\mathcal M}(L_{[\Sigma,\lambda,t]}).\)
\end{proposition}
\begin{proof}
Let $\Sigma = \{0, 1\}$,  $\lambda = \frac23$, and $t = \frac32$. Let ${\mathcal M} = (\{s\},\Sigma, s,\mathit{Act},P, \ell)$ be the single-state MDP with $\mathit{Act}=\{act_0,act_1\}$, $P(s,act_i,s)=1$, and $\ell(s,act_i)=i$ for $i\in\{0,1\}$.
By the proof of \cite[Proposition~1]{chatterjee_multi-objective_2013}, there is a word $w\in\Sigma^\omega$ with $\DS(w)=3/2$, but no eventually periodic word has discounted sum $3/2$. Hence $L_{[\Sigma,\lambda,t]}\neq\emptyset$, while $L^{<\infty}_{[\Sigma,\lambda,t]}$ contains no eventually periodic word. By Theorem~\ref{thm:omega-regularity}, $L^{<\infty}_{[\Sigma,\lambda,t]}$ is $\omega$-regular. Since every non-empty $\omega$-regular language contains an eventually periodic word, we obtain $L^{<\infty}_{[\Sigma,\lambda,t]}=\emptyset$.
Therefore, $V^{\sup}_{\mathcal M}(L^{<\infty}_{[\Sigma,\lambda,t]})=0$, and Lemma~\ref{lem:MDP-inf} gives $V^{\FM-\sup}_{\mathcal M}(L_{[\Sigma,\lambda,t]})=0$.
On the other hand, $M$ has a path for every word in $\Sigma^\omega$. Choose $w\in L_{[\Sigma,\lambda,t]}$ and let $\alpha^*$ deterministically follow the corresponding action sequence. Then $w$ is generated with probability $1$, and hence $V^{\alpha^*}_{\mathcal M}(L_{[\Sigma,\lambda,t]})=V^{\sup}_{\mathcal M}(L_{[\Sigma,\lambda,t]})=1$. Thus $V^{\alpha^*}_{\mathcal M}(L_{[\Sigma,\lambda,t]})>V^{\FM-\sup}_{\mathcal M}(L_{[\Sigma,\lambda,t]})$, as required.
\end{proof}

\subparagraph*{Remark} Our results on MDPs with \TDS objectives extend two decidability results for deterministic $1$-player games --a special case of MDPs-- represented by finite edge-labelled graphs. In such graphs, deterministic finite-memory strategies generate exactly the eventually periodic paths. Thus, the finite-memory supremum is $1$ exactly when an eventually periodic path reaches the target, and $0$ otherwise. This is a special case of the constrained eventually-periodic \TDS problem, shown decidable by Boker et al.~\cite{boker_target_2015}. The infimum, in turn, is $0$ exactly when some path has discounted sum different from the target, and $1$ otherwise. Such a path exists precisely when the minimum discounted sum is below the target or the maximum is above it, which can be checked using standard discounted-sum optimization in MDPs~\cite{Puterman94}.

\subsection{\TDS objectives in stochastic games}
\label{sec:SG-result}

\subparagraph*{Preliminaries}
A \emph{stochastic game} (SG) is a tuple
$\mathcal G=(S,\Sigma,s_{\mathsf{init}},\mathit{Act},P,\ell,\langle S_\Max,S_\Min\rangle)$,
where $(S,\Sigma,s_{\mathsf{init}},\mathit{Act},P,\ell)$ is an MDP and
$\langle S_\Max,S_\Min\rangle$ is a partition of $S$ into states controlled by the
\emph{maximizer} $\Max$ and by the \emph{minimizer} $\Min$. Histories are defined as for MDPs, and $\mathrm H^{\Max}_{\mathcal G}$ (resp.\ $\mathrm H^{\Min}_{\mathcal G}$) denotes the set of
histories whose last state belongs to $S_\Max$ (resp.\ $S_\Min$). A \emph{strategy of $\Max$}
is a function $\alpha\colon\mathrm H^{\Max}_{\mathcal G}\to\mathrm{Dist}(\mathit{Act})$
with $\alpha(h)(a)=0$ whenever $a$ is not enabled at the last state of $h$; strategies
$\beta$ of $\Min$ are defined symmetrically on $\mathrm H^{\Min}_{\mathcal G}$. Finite
memory is defined as for MDPs, the memory-update function being applied to histories
of the respective player. We write $\Theta^{\Max}_{\mathcal G},\Theta^{\Min}_{\mathcal G}$
for the sets of strategies of $\Max$ and $\Min$, and
$\Theta^{\Max,\FM}_{\mathcal G},\Theta^{\Min,\FM}_{\mathcal G}$ for their finite-memory
subsets.

A strategy profile $(\alpha,\beta)\in\Theta^{\Max}_{\mathcal G}\times\Theta^{\Min}_{\mathcal G}$ 
resolves all non-deterministic choices and induces a (possibly countably infinite) Markov chain
$\mathcal C^{\mathcal G,\alpha,\beta}$ over $\Sigma$, which is finite whenever both
$\alpha$ and $\beta$ are finite memory. As for MDPs, the objective of $\Max$ is to maximize the probability of a measurable language $L\subseteq\Sigma^\omega$, and we let
$V^{\alpha,\beta}_{\mathcal G}(L):=\Pr_{\mathcal C^{\mathcal G,\alpha,\beta}}(L)$
\[
  V^{\sup}_{\mathcal G}(L):=\sup_{\alpha\in\Theta^{\Max}_{\mathcal G}}\ \inf_{\beta\in\Theta^{\Min}_{\mathcal G}}V^{\alpha,\beta}_{\mathcal G}(L),
  \qquad
  V^{\FM\text{-}\sup}_{\mathcal G}(L):=\sup_{\alpha\in\Theta^{\Max,\FM}_{\mathcal G}}\ \inf_{\beta\in\Theta^{\Min}_{\mathcal G}}V^{\alpha,\beta}_{\mathcal G}(L),
\]
that is, the best probability that $\Max$ can guarantee against every strategy of $\Min$,
using arbitrary and finite memory respectively. A strategy $\alpha\in\Theta^{\Max}_{\mathcal G}$ is \emph{$\sup$-optimal} for $L$ if $\inf_{\beta}V^{\alpha,\beta}_{\mathcal G}(L)=V^{\sup}_{\mathcal G}(L)$, and $\alpha\in\Theta^{\Max,\FM}_{\mathcal G}$ is \emph{$\FM\text{-}\sup$-optimal} for $L$ if $\inf_{\beta}V^{\alpha,\beta}_{\mathcal G}(L)=V^{\FM\text{-}\sup}_{\mathcal G}(L)$. We write $|\mathcal G|$ for the size of the underlying MDP.

The game $\mathcal G$ is an MDP if $S_\Min = \emptyset$.  It is a \emph{(deterministic) $2$-player game} if $P(s,a,s')\in\{0,1\}$ for all $s,s'\in S$ and $a\in\mathit{Act}$, that is, if every action enabled at a state has a unique successor. In that case a deterministic strategy profile $(\alpha,\beta)$ induces a single play and $V^{\alpha,\beta}_{\mathcal G}(L)\in\{0,1\}$.

Finally, fixing a finite-memory strategy $\alpha\in\Theta^{\Max,\FM}_{\mathcal G}$
implemented by $(Q,q_{\mathsf{init}},\delta,\tau)$ yields a finite MDP
$\mathcal G^{\alpha}$ over $\Sigma$,
whose strategies are in memory-preserving bijection with $\Theta^{\Min}_{\mathcal G}$ and
satisfy $V^{\beta}_{\mathcal G^{\alpha}}(L)=V^{\alpha,\beta}_{\mathcal G}(L)$ for every
measurable objective $L$.

\paragraph*{Optimal values and strategies} \cref{thm:MDP-dec} generalizes to stochastic games as follows.

\begin{theorem}\label{thm:SG-dec}Let $\Sigma\subset\Z$ be a finite alphabet, with $\mu=\max_{\sigma\in\Sigma}|\sigma|$, let $\lambda = p/q\in(0,1)\cap\Q$ a discount factor,
$t=a/b\in\Q$ a target value, with $a/b$ and $p/q$ in reduced form, and let $\mathcal G$ be a finite SG over $\Sigma$.  
For the \TDS objective $L_{[\Sigma,\lambda,t]}$
\begin{itemize}
\item the $\FM$-$\sup$-optimal value in $\mathcal{G}$ can be computed 
\item an $\FM$-$\sup$-optimal deterministic strategy using $O(b q^2 \mu)$ memory can be synthesized
\item deciding whether the $\FM$-$\sup$ optimal value exceeds a given rational threshold is in $\mathrm{NP}\cap\mathrm{coNP}$ when $b, q$ and $\mu$ are given in unary. 
\end{itemize}
Moreover, these tasks can be carried out in time polynomial in $|\mathcal G| \cdot |\Sigma|\cdot b q^2 \mu$ when $\mathcal G$ is a $2$-player deterministic game.
\end{theorem}
\begin{proof}[Proof sketch]
The proof rests on the identity
\begin{equation}\label{eq:sg-reduction}
  V^{\FM\text{-}\sup}_{\mathcal G}\bigl(L_{[\Sigma,\lambda,t]}\bigr)
  \;=\;V^{\sup}_{\mathcal G}\bigl(L^{<\infty}_{[\Sigma,\lambda,t]}\bigr),
\end{equation}
which is proved in Appendix~\ref{app:proof}, and reduces the optimization of \TDS objectives for finite-memory strategies of $\Max$ to the analogous problem for the sub-language $L^{<\infty}_{[\Sigma,\lambda,t]}$.

By Item 1 of~\cref{thm:omega-regularity}, $L^{<\infty}_{[\Sigma,\lambda,t]}=L(\mathcal A)$ for a deterministic safety automaton $\mathcal A$ with $|\mathcal A|=O(bq^2\mu)$, that one can construct effectively in time $O(bq^2\mu\cdot|\Sigma|)$. To conclude, it remains to reuse know n results on safety objectives in stochastic games~\cite{Fijalkow_2026}: writing $v:=V^{\sup}_{\mathcal G}(L^{<\infty}_{[\Sigma,\lambda,t]})$, the value $v$ is computable and a deterministic $\sup$-optimal strategy $\alpha^*\in\Theta^{\Max,\FM}_{\mathcal G}$ for $L^{<\infty}_{[\Sigma,\lambda,t]}$ using memory of size $|\mathcal A|$ can be synthesized, while comparing $v$ with a rational threshold is in $\mathrm{NP}\cap\mathrm{coNP}$. 

When $\mathcal G$ is a deterministic $2$-player game, the optimizing the \TDS objective reduces to solving a deterministic $2$-player safety game. Hence $v\in\{0,1\}$, and the value $v$ and an optimal stategy $\alpha^*$ can be computed in time polynomial in $|\mathcal G|\cdot|\mathcal A|$~\cite{Fijalkow_2026}. As $|\mathcal A|=O(bq^2\mu)$, this bound are pseudo polynomial and $\alpha^*$ uses pseudo-polynomial memory.
\end{proof}

\section{Conclusion}
In this paper, we solve the stochastic target discounted-sum problem and address further algorithmic questions about the distribution of discounted-sum payoffs in Markov chains and target discounted-sum objectives in stochastic games.
The original (deterministic) target discounted-sum problem remains open. Our contributions also leave some questions about the stochastic setting unresolved.
For probabilities of discounted-sum threshold languages in Markov chains, we provide  an approximation algorithm (up to any precision) but its computational complexity remains to be determined. We also leave open whether these probabilities are always rational and, if so, whether they can be computed exactly. 
Finally, while this paper focused on target discounted-sum in stochastic system, we are pursuing the opening that Theorem~\ref{thm:MC-threshold} provides to solving similar problems on threshold discounted-sum in stochastic systems (other than Markov chains).

\subparagraph*{Acknowledgments}
We thank Mohan Dantam and Pavol Kebis for their feedback on an earlier
version of this paper. We additionally thank Mohan Dantam for his observation
about Lemma~\ref{Lemma:MC-SC-1}, which prompted the statement and proof of
Theorem~\ref{thm:atoms}.


\bibliography{references-arXiv}
\appendix
\section{Proofs}\label{app:proof}
\begin{lemma}\label{lem:measurable}
$L_{[\Sigma,\lambda,t]}$, $L^{\infty}_{[\Sigma,\lambda,t]}$ and
$L^{<\infty}_{[\Sigma,\lambda,t]}$ are measurable.
\end{lemma}

\begin{proof}
For $i\ge 0$ the suffix sum function $S_i(w):=\DS(w[i\dots\infty])$ is a random variable on
$\Sigma^\omega$(follows from ~\cite[Chap.~5]{Puterman94}). In particular so are
$\DS=S_0$ and $S_i-S_j$ for all $i,j\ge0$.
Hence $L_{[\Sigma,\lambda,t]}=\DS^{-1}(\{t\})$ and
$U_{i,j}:=\{w\mid S_i(w)\neq S_j(w)\}=(S_i-S_j)^{-1}(\R\setminus\{0\})$ are measurable.
Since $\{S_i(w)\mid i\ge0\}$ is infinite iff for every $n$ some $S_i(w)$ differs from
all of $S_0(w),\dots,S_n(w)$, we have
\[
  L^{\infty}_{[\Sigma,\lambda,t]} \;=\; L_{[\Sigma,\lambda,t]} \,\cap
  (\bigcap_{n\ge0}\ \bigcup_{i\ge0}\ \bigcap_{j=0}^{n} U_{i,j}),
\]
a countable combination of measurable languages, hence measurable. Also $L^{<\infty}_{[\Sigma,\lambda,t]}$ is measurable as
$L^{<\infty}_{[\Sigma,\lambda,t]}=L_{[\Sigma,\lambda,t]}\setminus L^{\infty}_{[\Sigma,\lambda,t]}$.
\end{proof}

\begin{proof}[Proof of Lemma~\ref{lemma:Q}]
Let $w = d_0 d_1 \dots \in L^{<\infty}_{[\Sigma,\lambda,t]}$. By the definition of the language, $\DS(w) = t$, and the set of suffix-sums $S = \{\DS(w[i..\infty]) \mid i \ge 0\}$ is finite. Assume without loss of generality that the discount factor $\lambda = \frac{p}{q}$ and the target $t =\frac ab$ are in reduced form; that is, $\gcd(p,q) = 1$ and $\gcd(a,b) = 1$.
Because the elements of $\Sigma$ are bounded by $d_{\min}$ and $d_{\max}$, the inequality
\[
\frac{d_{\min}}{1-\lambda} \le \DS(w[i..\infty]) \le \frac{d_{\max}}{1-\lambda}
\]
holds for all $i \ge 0$. Thus, to prove that $\DS(w[i..\infty]) \in Q_{[\Sigma,\lambda,t]}$, it remains only to show that $\DS(w[i..\infty])$ can be written in the form $\frac{x}{bp}$ for some integer $x$.
Since $\DS(w) = t$ and, for every $i \ge 0$, $\DS(w) = \DS(w[0\dots i-1]) + \lambda^i\DS(w[i..\infty])$, the number $\DS(w[i..\infty])$ is a rational. Let $\DS(w[i..\infty]) = \frac{u_i}{v_i}$, where $u_i \in \mathbb{Z}$, $v_i \in \mathbb{Z}^+$, and $\gcd(u_i, v_i) = 1$. We wish to show that for all $i \ge 0$, the denominator $v_i$ divides the product $bp$.
For a prime number $f$ and a integer $n$, let $\text{pow}_f(n)$ denote the exponent of the highest power of $f$ that divides $n$, with the convention that $\text{pow}_f(0) = +\infty$.
Assume for contradiction that there exists some index $k \ge 0$ such that $v_k$ does not divide $bp$. This assumption implies that there exists a prime number $f$ such that
\[
\text{pow}_f(v_k) > \text{pow}_f(bp).
\]
Because the set of suffix-sums $S$ is finite, the set of denominators $\{v_i \mid i \ge 0\}$ must also be finite; therefore, the maximum value $M = \max_{i \ge 0} \text{pow}_f(v_i)$ exists. By our assumption, $M > \text{pow}_f(bp)$. Let $k$ be the smallest index such that $\text{pow}_f(v_k) = M$.
We immediately observe that
\begin{equation}\label{eq:Q-lemma-obs}
    \text{pow}_f(v_k) = M > \text{pow}_f(bp) \geq 0.
\end{equation}
Moreover, because $\gcd(u_k, v_k) = 1$ and $f$ divides $v_k$  (by the above equation), it follows that $f$ does not divide $u_k$.
\begin{equation}\label{eq:Q-lemma-case1}
     \text{pow}_f(u_k) = 0.
\end{equation}
We analyze this situation using two exhaustive cases.

\textbf{Case 1:} The prime $f$ divides $p$.
We will derive a contradiction by proving that $\text{pow}_f(v_{k+1}) > \text{pow}_f(v_k)$, which contradicts the choice of $k$.

Because $\DS(w[k..\infty]) = d_k + \lambda \DS(w[{k+1}..\infty])$, rearranging the terms and substituting the fractional forms gives the equation
\[
\frac{u_{k+1}}{v_{k+1}} = \frac{q}{p} \left( \frac{u_k}{v_k} - d_k \right) = \frac{q u_k - q d_k v_k}{p v_k}.
\]
Let $N = q u_k - q d_k v_k$ and $D = p v_k$. The exponent of $f$ in the reduced denominator $v_{k+1}$ is given by the difference of the exponents in the unreduced fraction:
\[
\text{pow}_f(v_{k+1}) = \text{pow}_f(D) - \text{pow}_f(N).
\]
Observe that
\[
\text{pow}_f(D) = \text{pow}_f(p) + \text{pow}_f(v_k).
\]
Next, we show that $\text{pow}_f(N) = 0$ by analyzing the exponents of the two terms in $N$. First, because $\gcd(p,q) = 1$, the assumption that $f$ divides $p$ implies that $f$ does not divide $q$, that is, $\text{pow}_f(q) = 0$.
Thus together with Equation~\eqref{eq:Q-lemma-case1}, we have $\text{pow}_f(q u_k) = 0$. Second, we see that $\text{pow}_f(q d_k v_k) > 0$ because $\text{pow}_f(q d_k v_k) \ge \text{pow}_f(v_k)$ and $\text{pow}_f(v_k) > 0$ (Equation~\eqref{eq:Q-lemma-obs}).
Because the two terms in $N$ have distinct exponents of $f$, the exponent of their difference is the minimum of the two exponents:
\begin{align*}
    \text{pow}_f(N) &= \text{pow}_f(q u_k - q d_k v_k) \\
    &= \min\{\text{pow}_f(q u_k), \text{pow}_f(q d_k v_k)\}  = 0.
\end{align*}

Substituting $\text{pow}_f(D)$ and $\text{pow}_f(N)$ into the equation for $\text{pow}_f(v_{k+1})$ gives
\[
\text{pow}_f(v_{k+1}) = \text{pow}_f(D) - \text{pow}_f(N) = \text{pow}_f(p)+ \text{pow}_f(v_k).
\]
Because $f$ divides $p$, we know that $\text{pow}_f(p) \ge 1$. Hence, $\text{pow}_f(v_{k+1}) > \text{pow}_f(v_k)$, which contradicts the maximality of $M$.

\textbf{Case 2:} The prime $f$ does not divide $p$.
We will derive a contradiction by showing that $k \ge 1$ and $\text{pow}_f(v_{k-1}) = \text{pow}_f(v_k)$. This contradicts the assumption that $k$ is the smallest index for which $\text{pow}_f(v_k) = M$.
We first show that $M >\text{pow}_f(v_0)$ to conclude that $k \geq 1$. 
Observe that $\frac{u_0}{v_0} = \DS(w) = t = \frac{a}{b}$. Because both $\frac{u_0}{v_0}$ and $\frac{a}{b}$ are in reduced form, $v_0 = b$, therefore $\text{pow}_f(v_0) = \text{pow}_f(b)$. Moreover, $M >\text{pow}_f(b)$ because $M > \text{pow}_f(bp)$ (Equation~\eqref{eq:Q-lemma-obs}). Therefore, it follows that $M > \text{pow}_f(v_0)$. 
Next we show that $\text{pow}_f(v_{k-1}) = \text{pow}_f(v_k)$. 
Because $\DS(w[k-1..\infty]) = d_{k-1} + \lambda \DS(w[k..\infty])$, substituting the fractional forms gives the equation
\[
\frac{u_{k-1}}{v_{k-1}} = {d_{k-1}} + \frac{p}{q} \frac{u_k}{v_k} = \frac{q d_{k-1} v_k + p u_k}{q v_k}.
\]
Let $N = q d_{k-1} v_k + p u_k$ and $D = q v_k$. Then
\[
\text{pow}_f(v_{k-1}) = \text{pow}_f(D) - \text{pow}_f(N).
\]
Observe that
\[
\text{pow}_f(D) = \text{pow}_f(q) + \text{pow}_f(v_k).
\]
Next, we have $\text{pow}_f(N) = 0$. This holds because $\text{pow}_f(p u_k) = 0$ (since $f$ does not divide $p$, that is $\text{pow}_f(p) = 0$ and using Equation~\eqref{eq:Q-lemma-case1}) while $\text{pow}_f(q d_{k-1} v_k) > 0$ (by Equation~\eqref{eq:Q-lemma-obs}).
Substituting $\text{pow}_f(D)$ and $\text{pow}_f(N)$ into the equation for $\text{pow}_f(v_{k-1})$ yields
\[
\text{pow}_f(v_{k-1}) = \text{pow}_f(D) - \text{pow}_f(N) = \text{pow}_f(q) + \text{pow}_f(v_k).
\]
Note that $\text{pow}_f(v_{k-1}) = M$ because $\text{pow}_f(v_k) = M$, $\text{pow}_f(v_{k-1}) \leq M$, and $\text{pow}_f(q) \geq 0$. This contradicts the assumption that $k$ is the smallest index for which $\text{pow}_f(v_k) = M$. 
\end{proof}
\begin{proof}[Proof of \cref{lemma:Automata-correctness}]
The direction $L^{<\infty}_{[\Sigma,\lambda,t]} \subseteq L(\mathcal{A}_{[\Sigma,\lambda,t]})$ is straightforward. For a word $w = d_0d_1 \dots \in L^{<\infty}_{[\Sigma,\lambda,t]}$, the sequence $q_0, q_1,\dots$ where $q_i = \DS(w[i..\infty])$ is its unique run in $\mathcal{A}_{[\Sigma,\lambda,t]}$ that avoids $\bot$. Indeed, every $q_i \in Q_{[\Sigma,\lambda,t]}$ (by Lemma~\ref{lemma:Q}), the initial state $q_0 =t$, and every pair of states satisfy the transition relation $q_i = d_i + \lambda q_{i+1}$.
We now show the other direction $L(\mathcal{A}_{[\Sigma,\lambda,t]}) \subseteq L^{<\infty}_{[\Sigma,\lambda,t]}$. Let $w = d_0 d_1 \dots \in L(\mathcal{A}_{[\Sigma,\lambda,t]})$. Because $\mathcal{A}_{[\Sigma,\lambda,t]}$ is a deterministic safety automaton with the set of safe states $F = Q_{[\Sigma,\lambda,t]}$, the unique run $q_0, q_1, \dots$ of the automaton on $w$ avoids the sink state $\bot$. Thus, for all $i \ge 0$, we have $q_i \in Q_{[\Sigma,\lambda,t]}$.
By the definition of the transition function $\delta$, we have, for all $i \ge 0$:
\(
q_i = d_i + \lambda q_{i+1}.
\)
By repeated substitution, this gives, for all $n \ge 0$:
\begin{equation} \label{eq:lemma-correctness-autoamta}
    q_i = \sum_{j=0}^{n} d_{i+j} \lambda^j + \lambda^{n+1} q_{i+ n+1}.
\end{equation}
Since $Q_{[\Sigma,\lambda,t]}$ is a finite set, the values $q_{i+n+1}$ are bounded. Thus, the reminder term $\lambda^{n+1} q_{i+n+1}$ vanishes as $n \to \infty$. Furthermore, because rationals $d_{i+j}$ belongs to $\Sigma$ which is finite, the infinite series
$\sum_{j=0}^{\infty} d_{i+j} \lambda^j$ converges. Taking the limit as $n \to \infty$ in Equation~\eqref{eq:lemma-correctness-autoamta} therefore gives:
\(
q_i = \sum_{j=0}^{\infty} d_{i+j} \lambda^j  = \DS(w[i\dots \infty]).
\)
Since $q_0 = t$, it follows that the discounted sum $\DS(w) = t$. Moreover because $q_i \in Q_{[\Sigma,\lambda,t]}$ for all $i \ge 0$, and $Q_{[\Sigma,\lambda,t]}$ is finite, the set of all suffix-sums $\{\DS(w[i..\infty]) \mid i \ge 0\}$ is also finite. Consequently, $w \in L^{<\infty}_{[\Sigma,\lambda,t]}$.
\end{proof}
\begin{proof}[Proof of \cref{thm:omega-regularity} (continued)]
    We establish that the deterministic safety automaton $\mathcal{A}_{[\Sigma,\lambda,t]}$ has a pseudo-polynomial size. Observe from Lemma~\ref{lemma:Q} that the cardinality of the state space $Q_{[\Sigma,\lambda,t]} \cup \{\bot\}$ of $\mathcal{A}_{[\Sigma,\lambda,t]}$ is bounded by
    \(bp\frac{d_{\max} - d_{\min}}{1 -\lambda} + 1,\)
    where $d_{\min} = \min \Sigma$, $d_{\max} = \max \Sigma$, and $p, b \in \mathbb{Z}^+$ such that $\lambda = \frac{p}{q}$ and $t = \frac{a}{b}$ for some $q \in \mathbb{Z}^+$ and $a \in \mathbb{Z}$. 
      Substituting the fractional forms of $\lambda$ in the state space bound gives:
    \(bpq\frac{d_{\max} - d_{\min}}{q - p} + 1.\)
    Since $\lambda \in (0,1)$, we have $p < q$. The above term is therefore strictly bounded by the worst-case scenario where $p = q - 1$, which gives:
    \(b(q-1)q(d_{\max} - d_{\min}) + 1.\)
    Consequently, the state space size is polynomial when $b$, $q$, $d_{\max}$, and $d_{\min}$ are encoded in unary.
\end{proof}

\begin{proof}[Proof of Theorem~\ref{thm:MC-threshold} (continued)]
We use the notation of the algorithm in the main text.
For every $n\ge0$, define the events
$E_n^>:=\{\pi\in\Paths(C)\mid\DS(\pi[0..n])>t+r_n\}$
and
$E_n^<:=\{\pi\in\Paths(C)\mid\DS(\pi[0..n])<t-r_n\}$.
Each is a disjoint union of cylinders of finite paths with
$n$ transitions. Hence $A_n=\Pr_C(E_n^>)$ and
$B_n=\Pr_C(E_n^<)$ are rational and computable by enumeration.

By Equation~\eqref{eq:DS-recursive}, every path
$\pi\in\Paths(C)$ satisfies
$|\DS(\pi)-\DS(\pi[0..n])|\le r_n$.
Consequently, $E_n^>$ contains only paths with discounted sum
greater than $t$, and $E_n^<$ contains only paths with
discounted sum less than $t$. Thus
\[
A_n\le
\Pr_C(\{w\in\Sigma^\omega\mid\DS(w)>t\})
\le 1-e-B_n.
\]
Whenever the stopping condition holds, these bounds show that
the returned value at most $\varepsilon$ far from the true probability.

It remains to prove termination.
Write $\pi=s_0s_1\cdots$.
Since $|\ell(s_n,s_{n+1})|\le\mu$ and
$r_n-r_{n+1}=\mu\lambda^n$, we have
\begin{align*}
\DS(\pi[0..n+1])-r_{n+1}
&\ge \DS(\pi[0..n])-r_n,\\
\DS(\pi[0..n+1])+r_{n+1}
&\le \DS(\pi[0..n])+r_n.
\end{align*}
Therefore $E_n^>\subseteq E_{n+1}^>$ and
$E_n^<\subseteq E_{n+1}^<$.
Moreover, if $\DS(\pi)>t$, then $r_n\to0$ allows us to
choose $n$ with $2r_n<\DS(\pi)-t$.
The tail bound then gives $\DS(\pi[0..n])>t+r_n$,
so $\pi\in E_n^>$.
The analogous argument applies when $\DS(\pi)<t$.
Hence
\[
\bigcup_{n\ge0}E_n^>
=\{\pi\in\Paths(C)\mid\DS(\pi)>t\},
\qquad
\bigcup_{n\ge0}E_n^<
=\{\pi\in\Paths(C)\mid\DS(\pi)<t\}.
\]
By continuity of probability for increasing unions,
$A_n$ and $B_n$ converge to the probabilities of discounted
sum greater than and less than $t$, respectively.
These probabilities sum to $1-e$, so
$1-e-A_n-B_n\to0$.
Since $\varepsilon>0$, the stopping condition is satisfied
at some finite depth $N$.
Adding the exact value $e$ or taking complements preserves
the precision, proving the remaining comparisons.
\end{proof}

\begin{proof}[Proof of \cref{thm:MDP-dec}(1)]
By Lemma~\ref{lem:MDP-inf}, the inf-optimal value
\begin{equation}\label{eq:appl-decide-inf-1}
    V^{\inf}_M(L_{[\Sigma,\lambda,t]}) = V^{\inf}_M(L^{<\infty}_{[\Sigma,\lambda,t]}).
\end{equation}
Theorem~\ref{thm:omega-regularity} establishes that the language $L^{<\infty}_{[\Sigma,\lambda,t]}$ is recognized by a deterministic safety automaton of pseudo-polynomial size. Hence, by standard product construction~\cite{PMC-book}, for objective $L^{<\infty}_{[\Sigma,\lambda,t]}$: (1) the inf-optimal value $V^{\inf}_M(L^{<\infty}_{[\Sigma,\lambda,t]})$ is computable in pseudo-polynomial time; and (2) a inf-optimal strategy $\alpha^*$ that is deterministic and uses pseudo-polynomial memory can be synthesized in pseudo-polynomial time. Thus by Equation~\eqref{eq:appl-decide-inf-1}, $V^{\inf}_M(L_{[\Sigma,\lambda,t]})$ is computable in pseudo-polynomial time.  We conclude the proof by showing $\alpha^*$ is also inf-optimal for objective $L_{[\Sigma,\lambda,t]}$, that is,  $V^{\inf}_M(L_{[\Sigma,\lambda,t]}) =V^{\alpha^*}_M(L_{[\Sigma,\lambda,t]})$.
First since $\alpha^*$ is inf-optimal for objective $L^{<\infty}_{[\Sigma,\lambda,t]}$ we have:
\begin{equation}
    \label{eq:appl-decide-inf-2}
    V^{\alpha^*}_M(L^{<\infty}_{[\Sigma,\lambda,t]}) = V^{\inf}_M(L^{<\infty}_{[\Sigma,\lambda,t]}).
\end{equation}
Next since $\alpha^*$ is a finite-memory strategy, it induces a finite Markov chain. By Theorem~\ref{thm:MC-measure-zero} we have:
$V^{\alpha^*}_M(L^{<\infty}_{[\Sigma,\lambda,t]}) = V^{\alpha^*}_M(L_{[\Sigma,\lambda,t]}).$
Thus, together with equations~\ref{eq:appl-decide-inf-1} and \ref{eq:appl-decide-inf-2}, $V^{\inf}_M(L_{[\Sigma,\lambda,t]}) =V^{\alpha^*}_M(L_{[\Sigma,\lambda,t]})$, which concludes the proof. 
\end{proof}
\begin{proof}[Proof of \cref{thm:MDP-dec}(2)]
By Lemma~\ref{lem:MDP-inf}, the FM-sup-optimal value
\begin{equation}\label{eq:appl-decide-sup-1}
    V^{\FM-\sup}_M(L_{[\Sigma,\lambda,t]}) = V^{\sup}_M(L^{<\infty}_{[\Sigma,\lambda,t]}).
\end{equation}
Theorem~\ref{thm:omega-regularity} establishes that the language $L^{<\infty}_{[\Sigma,\lambda,t]}$ is recognized by a deterministic safety automaton of pseudo-polynomial size. Hence, by standard product construction and the algorithms for solving safety objectives in MDP~\cite{PMC-book}, for objective $L^{<\infty}_{[\Sigma,\lambda,t]}$: (1) the sup-optimal value $V^{\sup}_M(L^{<\infty}_{[\Sigma,\lambda,t]}) $ is computable in pseudo-polynomial time; and (2) a sup-optimal strategy $\alpha^*$ that is deterministic and uses pseudo-polynomial memory can be synthesized in pseudo-polynomial time. 
Thus by Equation~\eqref{eq:appl-decide-sup-1}, $V^{\textsf{FM-}\sup}_M(L_{[\Sigma,\lambda,t]})$ is computable in pseudo-polynomial time.
We conclude the proof by showing $\alpha^*$ is also FM-sup-optimal for objective $L_{[\Sigma,\lambda,t]}$, that is,  $V^{\textsf{FM-}\sup}_M(L_{[\Sigma,\lambda,t]}) =V^{\alpha^*}_M(L_{[\Sigma,\lambda,t]})$.

First, since $\alpha^*$ is a sup-optimal strategy for objective $L^{<\infty}_{[\Sigma,\lambda,t]}$ we have: 
\begin{equation}
    \label{eq:appl-decide-sup-2}
    V^{\alpha^*}_M(L^{<\infty}_{[\Sigma,\lambda,t]}) = V^{\sup}_M(L^{<\infty}_{[\Sigma,\lambda,t]}).
\end{equation}
Next, because $\alpha^*$ is a finite-memory strategy we have:
\begin{equation}
    \label{eq:appl-decide-sup-3}
    V^{\alpha^*}_M(L_{[\Sigma,\lambda,t]}) \leq V^{\textsf{FM-}\sup}_M(L_{[\Sigma,\lambda,t]}).
\end{equation}
Finally, because $L^{<\infty}_{[\Sigma,\lambda,t]}$ is a subset of $L_{[\Sigma,\lambda,t]}$, we have:
\(V^{\alpha^*}_M(L^{<\infty}_{[\Sigma,\lambda,t]}) \le V^{\alpha^*}_M(L_{[\Sigma,\lambda,t]}).\) Thus, together with equations~\ref{eq:appl-decide-sup-1}, \ref{eq:appl-decide-sup-2}, and \ref{eq:appl-decide-sup-3}, we have: $V^{\textsf{FM-}\sup}_M(L_{[\Sigma,\lambda,t]}) =V^{\alpha^*}_M(L_{[\Sigma,\lambda,t]})$, which concludes the proof.
\end{proof}

\begin{proof}[Proof of \cref{thm:SG-dec} (continued)]
We here establish \eqref{eq:sg-reduction} and show that $\alpha^*$ is
$\FM\text{-}\sup$-optimal for $L_{[\Sigma,\lambda,t]}$. Both follow by fixing a
finite-memory strategy of $\Max$ and applying \cref{lem:MDP-inf} to the finite MDP left
for $\Min$. Indeed, fix $\alpha\in\Theta^{\Max,\FM}_{\mathcal G}$: since
$\mathcal G^{\alpha}$ is a finite MDP controlled by $\Min$, the first identity of
\cref{lem:MDP-inf} yields
\[
  \inf_{\beta}V^{\alpha,\beta}_{\mathcal G}\bigl(L_{[\Sigma,\lambda,t]}\bigr)
  =V^{\inf}_{\mathcal G^{\alpha}}\bigl(L_{[\Sigma,\lambda,t]}\bigr)
  =V^{\inf}_{\mathcal G^{\alpha}}\bigl(L^{<\infty}_{[\Sigma,\lambda,t]}\bigr)
  =\inf_{\beta}V^{\alpha,\beta}_{\mathcal G}\bigl(L^{<\infty}_{[\Sigma,\lambda,t]}\bigr),
\]
and taking the supremum over $\alpha\in\Theta^{\Max,\FM}_{\mathcal G}$ gives
$V^{\FM\text{-}\sup}_{\mathcal G}(L_{[\Sigma,\lambda,t]})
 =V^{\FM\text{-}\sup}_{\mathcal G}(L^{<\infty}_{[\Sigma,\lambda,t]})$. The latter equals
$v$: it is at most $v$ because
$\Theta^{\Max,\FM}_{\mathcal G}\subseteq\Theta^{\Max}_{\mathcal G}$, and at least $v$
because $\alpha^*$ has finite memory and is $\sup$-optimal for
$L^{<\infty}_{[\Sigma,\lambda,t]}$; this proves \eqref{eq:sg-reduction}. Finally, since
$L^{<\infty}_{[\Sigma,\lambda,t]}\subseteq L_{[\Sigma,\lambda,t]}$, we get
$\inf_{\beta}V^{\alpha^*,\beta}_{\mathcal G}(L_{[\Sigma,\lambda,t]})
 \ge\inf_{\beta}V^{\alpha^*,\beta}_{\mathcal G}(L^{<\infty}_{[\Sigma,\lambda,t]})=v$, so
$\alpha^*$ is $\FM\text{-}\sup$-optimal for $L_{[\Sigma,\lambda,t]}$. 
\end{proof}

\end{document}